\documentclass[preprint,journal]{vgtc}       

\ifpdf
  \pdfoutput=1\relax                   
  \usepackage{graphicx}                
  \DeclareGraphicsExtensions{.pdf,.png,.jpg,.jpeg} 
\else
  \usepackage{graphicx}                
  \DeclareGraphicsExtensions{.eps}     
\fi%

\graphicspath{{figures/}{pictures/}{images/}{./}} 

\usepackage{microtype}                 
\PassOptionsToPackage{warn}{textcomp}  
\usepackage{textcomp}                  
\usepackage{mathptmx}                  
\usepackage{times}                     
\usepackage{cite}                      
\usepackage{tabu}                      
\usepackage{booktabs}                  

\usepackage{wrapfig}
\usepackage{todonotes}
\usepackage{mathptmx}
\usepackage{graphicx}
\usepackage{times}
\usepackage{amsmath,amscd,amssymb, amsthm}
\usepackage{times}
\usepackage[abs]{overpic}
\usepackage{cases}
\usepackage{float}
\usepackage{caption}
\usepackage{subfloat}
\usepackage{subfig}
\usepackage{wrapfig}
\usepackage{placeins}

\onlineid{1066}

\vgtccategory{Research}
\vgtcpapertype{algorithm/technique}

\title{Mode Surfaces of Symmetric Tensor Fields:\\Topological Analysis and Seamless Extraction}

\author{Botong Qu, Lawrence Roy, Yue Zhang,\textit{Member, IEEE}, and Eugene Zhang,\textit{Senior Member, IEEE}}
\authorfooter{
\item
B. Qu is with the School of Electrical Engineering and Computer Science, Oregon State University. E-mail: qub@oregonstate.edu.
\item
L. Roy is with the School of Electrical Engineering and Computer Science, Oregon State University. E-mail: royl@eecs.oregonstate.edu.
\item
Y. Zhang is an Associate Professor with the School of Electrical Engineering and Computer Science, Oregon State University. E-mail: zhangyue@oregonstate.edu.
\item
E. Zhang is a Professor with the School of Electrical Engineering and Computer Science, Oregon State University. E-mail: zhange@eecs.oregonstate.edu.
}

\shortauthortitle{Qu \MakeLowercase{\textit{et al.}}: Mode Surfaces of Symmetric Tensor Fields: Topological Analysis and Seamless Extraction}

\abstract{Mode surfaces are the generalization of degenerate curves and neutral surfaces, which constitute 3D symmetric tensor field topology. Efficient analysis and visualization of mode surfaces can provide additional insight into not only degenerate curves and neutral surfaces, but also how these features transition into each other. Moreover, the geometry and topology of mode surfaces can help domain scientists better understand the tensor fields in their applications. Existing mode surface extraction methods can miss features in the surfaces. Moreover, the mode surfaces extracted from neighboring cells have gaps, which make their subsequent analysis difficult. In this paper, we provide novel analysis on the topological structures of mode surfaces, including a common parameterization of all mode surfaces of a tensor field using 2D asymmetric tensors. This allows us to not only better understand the structures in mode surfaces and their interactions with degenerate curves and neutral surfaces, but also develop an efficient algorithm to seamlessly extract mode surfaces, including neutral surfaces. The seamless mode surfaces enable efficient analysis of their geometric structures, such as the principal curvature directions. We apply our analysis and visualization to a number of solid mechanics data sets.
} 


\CCScatlist{ 
 \CCScat{K.6.1}{Management of Computing and Information Systems}%
{Project and People Management}{Life Cycle};
 \CCScat{K.7.m}{The Computing Profession}{Miscellaneous}{Ethics}
}

\teaser{
\centering
\subfloat[][$\mu = -0.9998$ and $\mu = \pm 1.0$]{\includegraphics[width=0.3\linewidth]{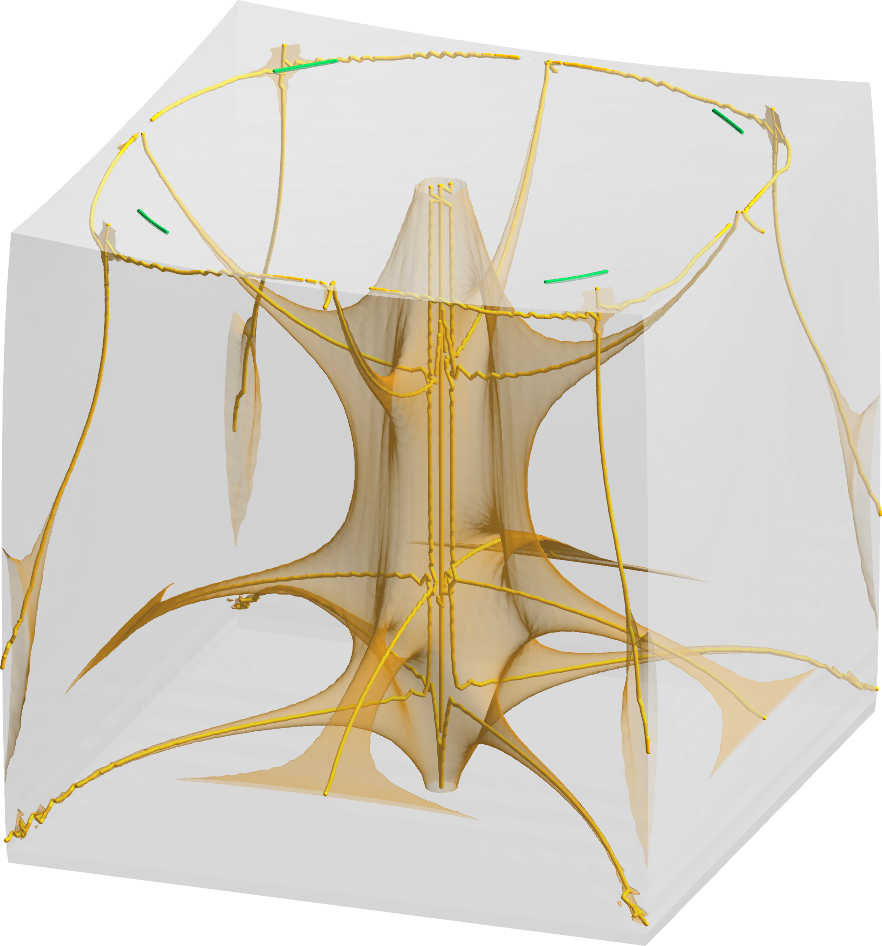}} \hspace{0.05in}
\subfloat[][$\mu = -0.80$]{\includegraphics[width=0.3\linewidth]{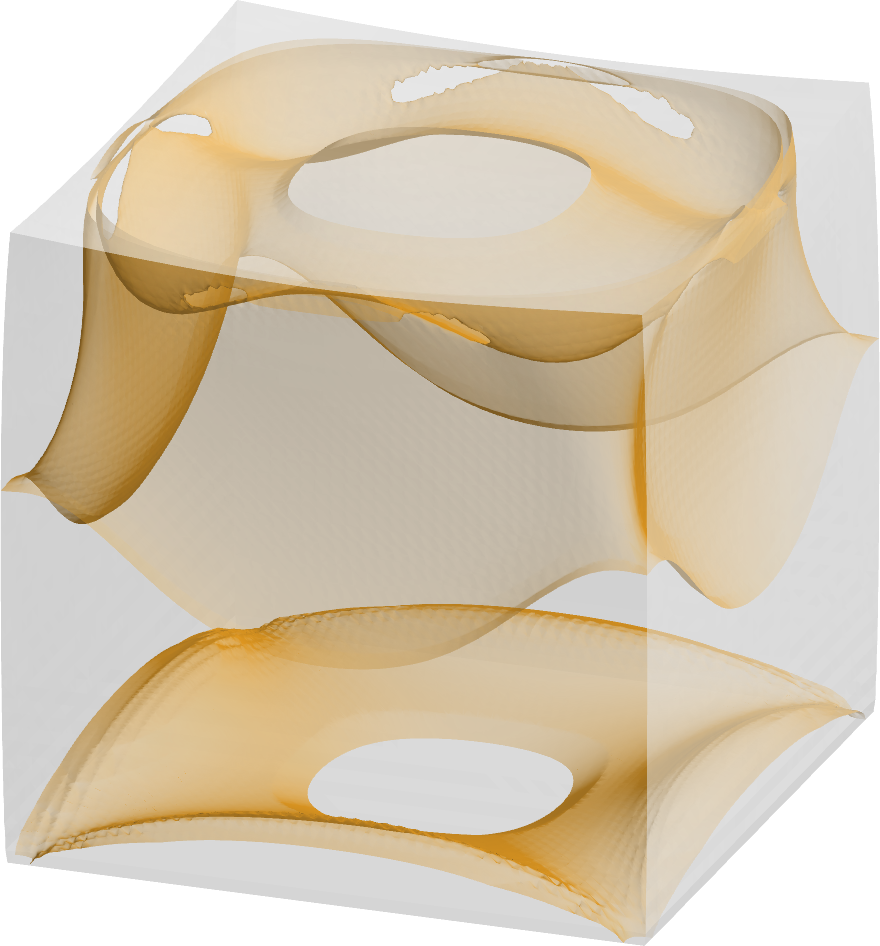}}\hspace{0.05in}
\subfloat[][$\mu = 0.0$ and $\mu = -0.45$]{\includegraphics[width=0.3\linewidth]{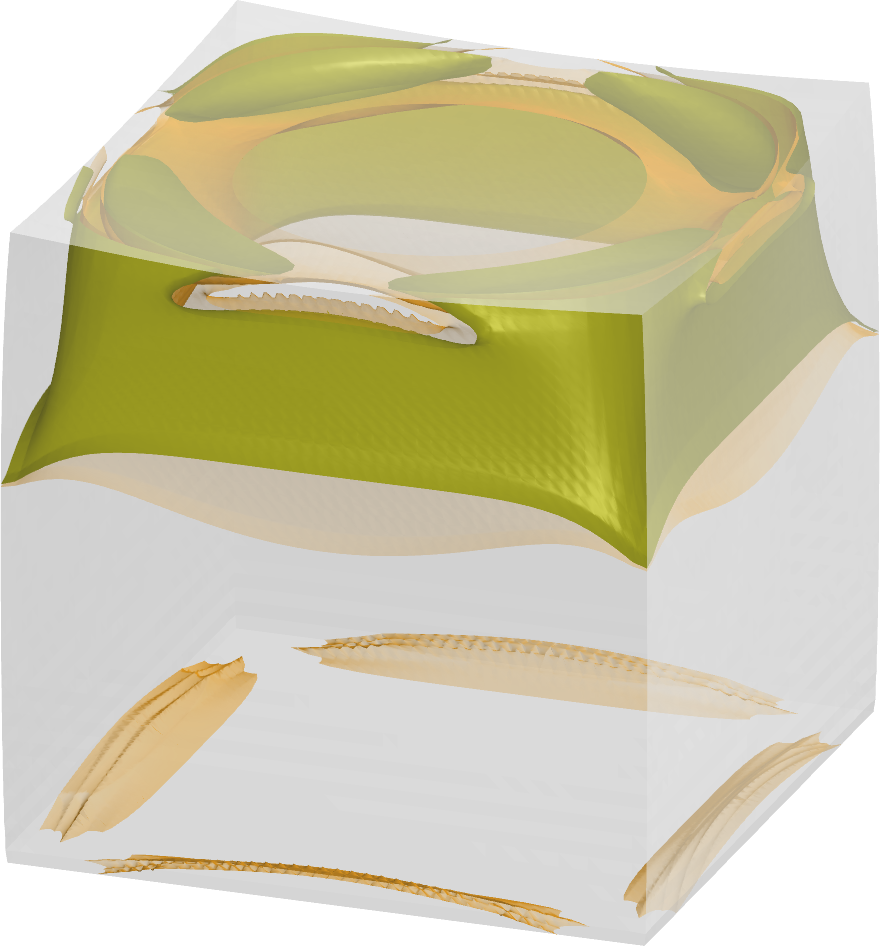}}
\caption{Mode surfaces of a stress tensor field for a block under compression. When the mode value $\mu$ is close to $\pm 1$ (a), the mode surface resembles a vascular structure that contains a core, the degenerate curves (yellow). On the other hand, when $\mu$ is close $0$ (c), the mode surface converges onto the neutral surfaces (chartreuse). By observing the change in mode surfaces' geometry and topology ((a)-(c)), we gain insight into the interaction between degenerate curves and neutral surfaces, the two constituents of tensor field topology. In addition, some mode surfaces can contain interesting features not present in degenerate curves and neutral surfaces, such as the bottom layer with a hole (b). Finally, the change in mode surfaces' topology when mode values change, such as the contraction of the mode surface from (b) to a vascular structure in (a) can also provide important insight into the underlying physics.  }
\label{fig:teaser}
}

\vgtcinsertpkg

\long\def\symbolfootnote[#1]#2{\begingroup%
\def\thefootnote{\fnsymbol{footnote}}\footnote[#1]{#2}\endgroup}

\newtheorem{theorem}{\sffamily Theorem}
\newtheorem{lemma}[theorem]{\sffamily Lemma}
\newtheorem{corollary}[theorem]{\sffamily Corollary}

\DeclareMathOperator{\trace}{trace}
\DeclareMathOperator{\adj}{adj}

\begin{document}


\firstsection{Introduction}
\label{sec:introduction}

\maketitle

Symmetric tensor fields have a wide range of applications in science, engineering, and medical domains.  The diffusion tensor field analysis in medical imaging plays a key role in diagnosing and treatment planning for brain cancers.  The stress and strain tensors in continuum mechanics  enable the predictions of structural failures.  Topology-driven analysis and visualization of 3D symmetric tensor fields have made much progress in recent years, of which the focus is on high-quality extraction and visualization of the tensor field topology such as {\em degenerate curves} (where the tensors have repeating eigenvalues) and {\em neutral surfaces} (where the major, medium, and minor eigenvalues of the tensors form an arithmetic sequence). On the other hand, degenerate curves and neutral surfaces are often treated as unrelated objects and interpreted separately.

In fact, both are a level set of the {\em mode} function of the tensor field~\cite{CRISCIONE:00}, which ranges from $-1$ to $1$, with neutral surfaces being the zeroth level set of this function and the degenerate curves being the $\pm 1$ level set. In solid mechanics~\cite{CRISCIONE:00}, a mode $-1$ tensor corresponds to {\em uniaxial compression} and a mode $1$ tensor corresponds to {\em uniaxial extension}. In contrast, a mode $0$ tensor corresponds to {\em pure shear}. At other mode values, we observe biaxial extension and compression.

Degenerate curves and neutral surfaces transition into each other, and understanding their interaction can provide more insight than interpreting them separately. Such an interaction can be better understood and visualized through level set surfaces whose iso-values are between $0$ and $\pm 1$. These surfaces are referred to as {\em mode surfaces} such as the gold-colored surfaces shown in Figure~\ref{fig:teaser}. Note how mode surfaces with mode values close to $\pm 1$ can form vascular structures around the degenerate curves (Figure~\ref{fig:teaser} (a): yellow curves). On the other hand, mode surfaces of values close to $0$ converge onto neutral surfaces (Figure~\ref{fig:teaser} (c): chartreuse surfaces). We refer the readers to our accompanying video for animations of the mode surfaces for this data set and other examples in the paper.

In addition, mode surfaces themselves can also provide an interesting insight into the underlying physics. For example, in Figure~\ref{fig:teaser} (b), the mode surface contains a sheet with a hole at the bottom of the volume. Note that such a feature is not present in degenerate curves and neutral surfaces. Therefore, it is important to study not only tensor field topology but also other feature surfaces such as mode surfaces.

The changes in the geometry and topology of mode surfaces as mode values change can also be meaningful of the underlying tensor field.  We observe the splitting of the bottom sheet in the mode surface (Figure~\ref{fig:teaser} (b)) into four disconnected components (Figure~\ref{fig:teaser} (c)) as well as the thinning of the mode surface (Figure~\ref{fig:teaser} (b)) into a vascular structure (Figure~\ref{fig:teaser} (a)).

Finally, it is often useful to not only visualize mode surfaces but also their differential properties such as surface normal and principal curvature directions. For example, in Figure~\ref{fig:quality_compare_IBFV} (b), we visualize the principal curvature directions using a texture-based method~\cite{Palacios:11}. The normal to a mode surface provides the direction in which the mode changes the most, e.g. from uniaxial extension to uniaxial compression and vice versa. As isosurfaces of the tensor mode function, the principal curvature directions of the mode surfaces are determined by the Hessian of the mode function~\cite{Johnson:04} and can provide information in the ridge and valley lines in the surfaces.

\begin{figure}[!tbh]
\centering%
\subfloat[][A-patches Method~\cite{Palacios:16}]{
\includegraphics[height=2.5in]{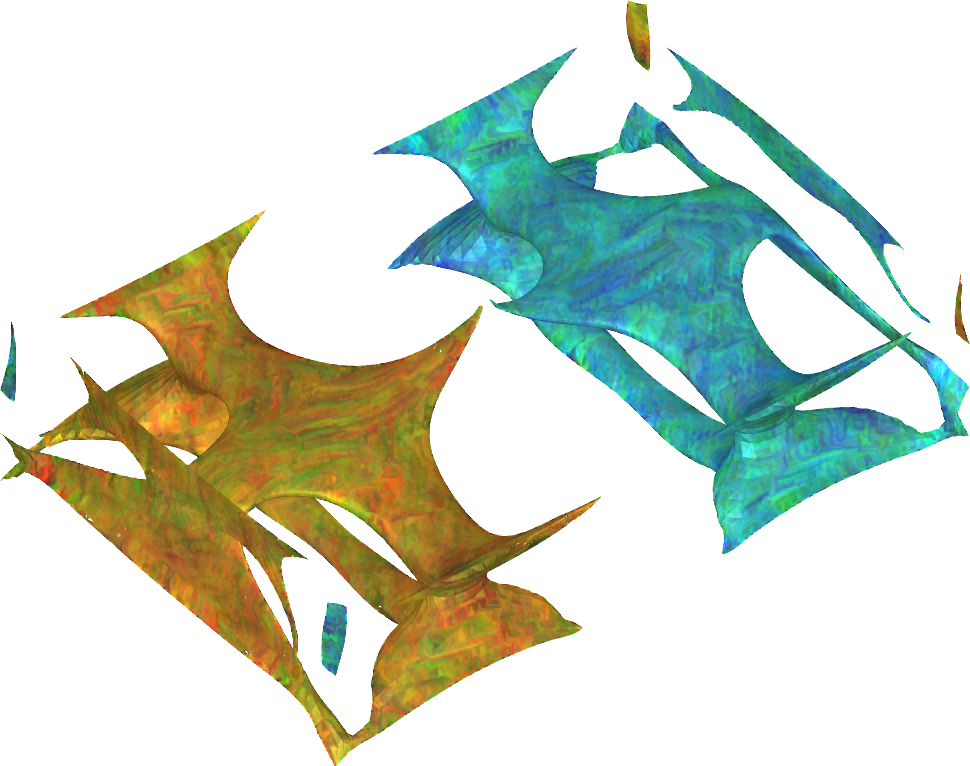}}\\
\subfloat[][Our method]{
\includegraphics[height=2.5in]{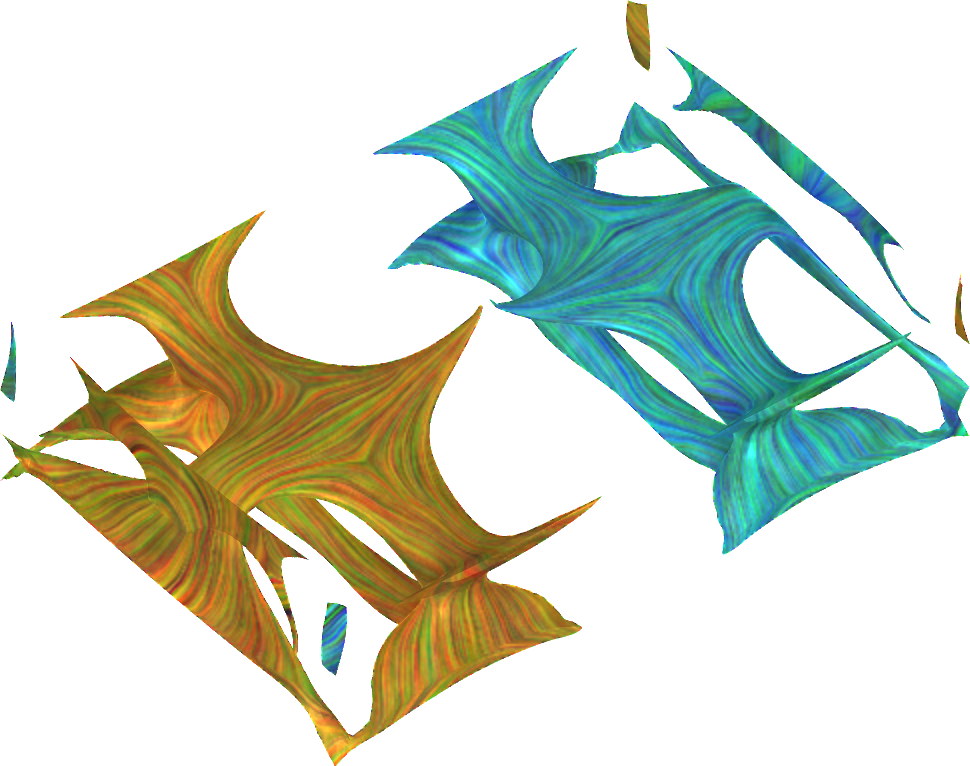}}
\caption[]{The holes and gaps in mode surfaces extracted using the A-patches method~\cite{Palacios:16} lead to a mesh with seams, impeding the computation of important surface properties such as normal and curvature (top: LIC texture lacks clarity in showing the major principal direction in the mode surface). Our method generates seamless meshes, which lead to robust curvature computation (bottom).
}
\label{fig:quality_compare_IBFV}
\end{figure}

Palacios et al.~\cite{Palacios:16} extract mode surfaces using the A-patches method~\cite{luk:2009}, which requires iterative subdivisions of tetrahedra in the mesh. Given a tetrahedron, the iterative subdivision process is not guaranteed to converge. Consequently, a maximum level of subdivision is used to stop the subdivision if the mode surface inside the tetrahedron cannot be completely extracted by the already performed subdivisions. This leads to holes in the extracted mode surface, which can be misinterpreted as the mode surface intersecting the domain boundary (Figure~\ref{fig:quality_compare}~(a) and (c)).

In addition, because the subdivision process is performed on each tetrahedron in the mesh separately, the intersection of the mode surfaces from adjacent tetrahedra (which are curves) are usually not consistent (different number of vertices on the shared face). This leads to gaps in the mode surfaces between adjacent tetrahedra. Instead of a connected mesh, the mode surface generated by Palacios et al.~\cite{Palacios:16} is essentially a triangle soup. While such gaps are not necessarily an issue when visualizing the mode surfaces, they present challenges in computing the normal and principal curvature directions on mode surfaces (as shown in Figure~\ref{fig:quality_compare_IBFV}~(a)).

\begin{figure}[!tbh]
\centering%
\subfloat[][The A-patches method~\cite{Palacios:16}]{
\includegraphics[height=1.2in]{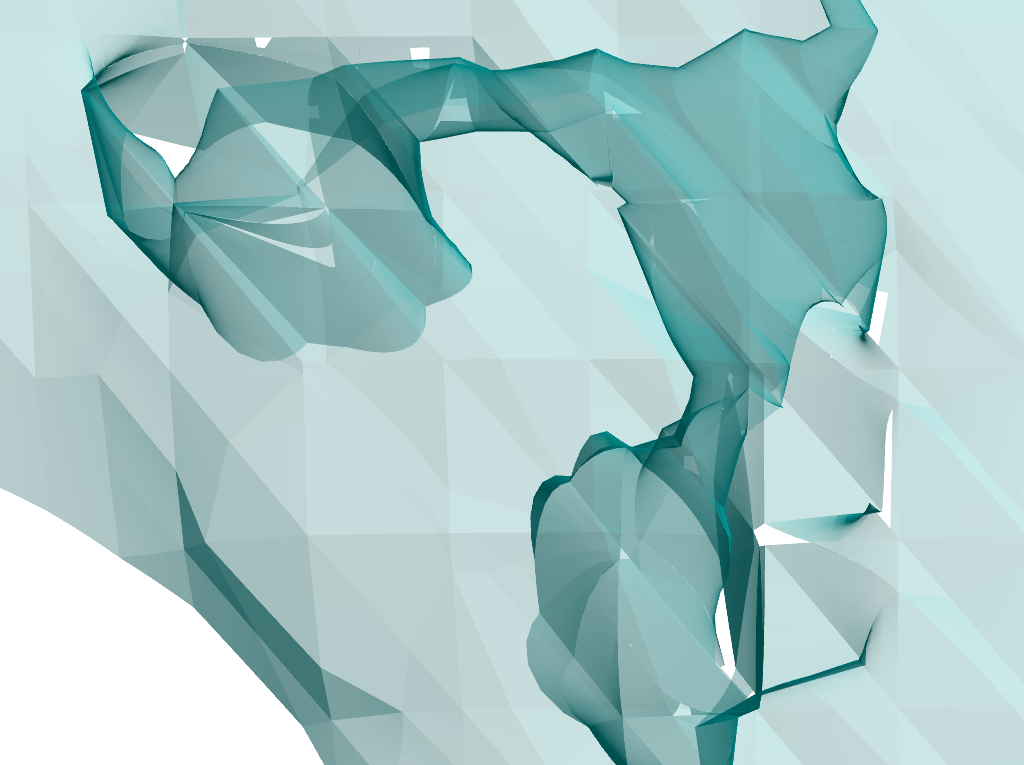}}
\hspace{0.05in}
\subfloat[][Our method]{
\includegraphics[height=1.2in]{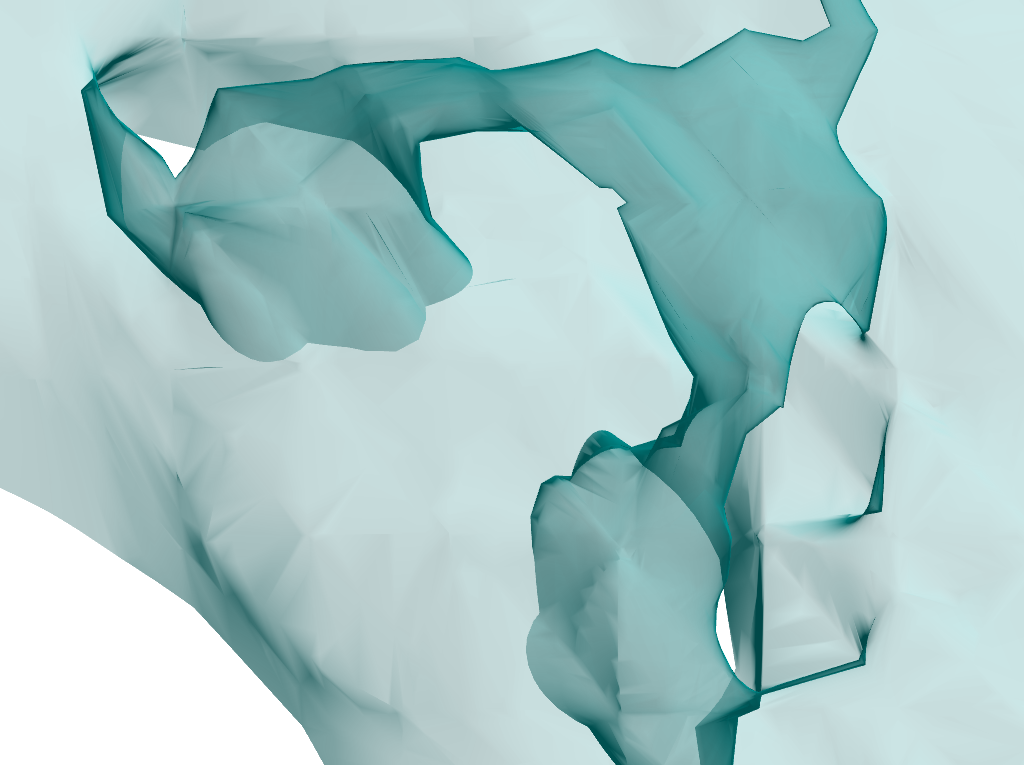}}
\hspace{0.05in}
\subfloat[][The method of~\cite{Roy:18}]{
\includegraphics[height=1.2in]{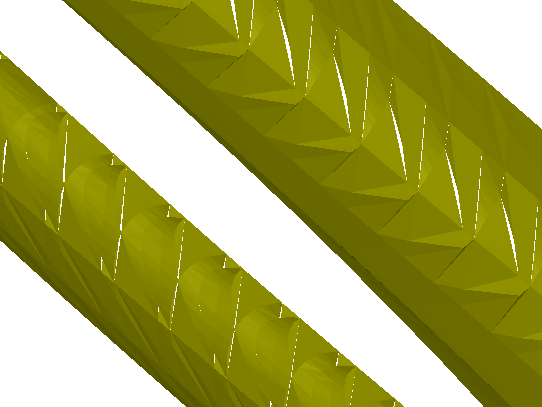}}\hspace{0.05in}
\subfloat[][Our method]{
\includegraphics[height=1.2in]{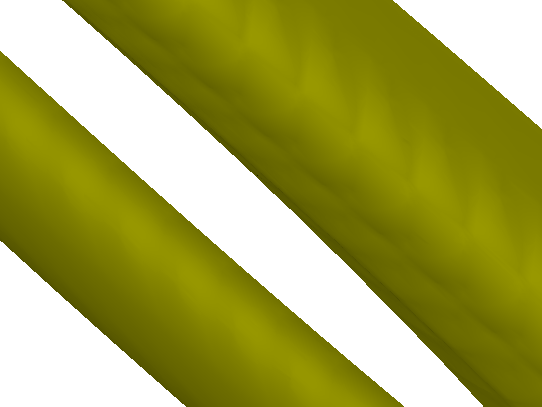}}\hspace{0.05in}
\caption[]{Existing mode surface extraction methods such as~\cite{Palacios:16} (top-left) can have holes due to non-convergence in the extraction process. In addition, these methods~\cite{Palacios:16,Roy:18} can have gaps between mode surfaces extracted from different tets (top-left and bottom-left). These holes and gaps can not only mislead interpretation but also make it difficult to compute differential properties of the mode surfaces such as their curvature tensors. Our method addresses these challenges with a unified framework in which mode surfaces (including neutral surfaces) can be extracted more accurately, faster, and in a seamless fashion (right column).
}
\label{fig:quality_compare}
\end{figure}

Moreover, the subdivision process incurs much computational cost, which can take seconds or even minutes to extract a single mode surface, depending on the time-quality tradeoff parameter (the maximum number of subdivisions).

In this paper, we address the aforementioned challenges by providing a parameterization for all mode surfaces given a 3D linear tensor field. This common parameterization allows us to extract mode surfaces at any accuracy without the need for mesh subdivision. This leads to fewer missing pieces in the extracted surfaces than from the A-patches method. Furthermore, our algorithm extracts mode surfaces inside each face in the mesh, which are then used to find mode surfaces inside each tetrahedron. This makes it straightforward to stitch mode surfaces from adjacent tetrahedra without gaps, resulting in a {\em seamless} mode surface on which differential properties can be computed and visualized (Figure~\ref{fig:quality_compare_IBFV} (b)). Moreover, a mode surface can be extracted under five seconds for our simulation data sets.

Our pipeline also applies to the extraction of neutral surfaces, which is a non-orientable surface~\cite{armstrong1979basic} characterized by a degree-three polynomial~\cite{Roy:18}. To our knowledge, the algorithm and the pipeline we develop for mode and neutral surface extraction, in conjunction with the degenerate curve extraction method of Roy et al.~\cite{Roy:18}, is the first {\em unified} framework for all isosurfaces of the mode function, including neutral surfaces (mode $0$) and degenerate curves (mode $\pm 1$).

To demonstrate the utility of our approach, we apply our tensor field analysis and visualization to solid mechanics applications.

\section{Related Work}\label{sec:related_work}

Tensor field visualization has advanced much in the last decades~\cite{EPFL-BOOK-138668,Kratz:13}.

Delmarcelle and Hesselink\cite{delmarcelle:visualizing} introduce the notion of topology for 2D symmetric tensor fields, which consists of degenerate points. The topological features of 3D symmetric tensor fields are first studied by Hesselink et al.~\cite{hesselink:topology}. Zheng and Pang~\cite{Zheng:04} point out that degenerate points form curves under structurally stable conditions, i.e.\ the structure persists under arbitrarily small perturbations~\cite{damon1998generic}. Several methods have been proposed to extract degenerate curves~\cite{Zheng:05a,Tricoche:08,Palacios:16,Roy:18}. Palacios et al.~\cite{Palacios:17} introduce editing operations for degenerate curves, such as degenerate curve removal, degenerate curve deformation, and degenerate curve reconnection.

Besides feature curves, another type of topological features is surfaces. Zobel and Scheuermann~\cite{Zobel2017} introduce the notion of extremal surfaces for 3D symmetric tensor fields. Raith et al.\cite{raith2019tensor} extract fiber surfaces of tensor fields by linearly interpolating tensor invariants in each tetrahedron. Palacios et al.~\cite{Palacios:16} introduce the notion of neutral surfaces.

Extracting implicit surfaces is a well-researched area~\cite{article}. The most popular technique, Marching Cubes~\cite{Lorensen:87} and its variants, focus on trilinear functions that are degree-three polynomials, while mode surfaces are of degree-six. Using surface extraction methods designed for a lower-degree polynomial, even with a guarantee for topological correctness~\cite{Nielson:04,Renbo:2005,Grosso:17}, can still miss important topological and geometric features of isosurfaces corresponding to a higher-degree polynomial. Implicit surface extraction methods that can handle more general functions with topological guarantees~\cite{Stander:97,BOISSONNAT2005405,Paiva06robustadaptive} usually require $C^2$ functions. Since the mode function in our case is a piecewise degree-six polynomial and $C^0$ at the cell boundaries, it is not clear how to adapt these techniques to the mode function.

The A-patches method~\cite{luk:2009} is a technique to extract algebraic surfaces, which is adapted by Palacios et al.~\cite{Palacios:16} to extract mode surfaces. However, to our knowledge, the extracted surfaces can contain seams at the faces of the mesh.

Roy et al.~\cite{Roy:18} provide algorithms to extract degenerate curves and neutral surfaces from linear tensor fields based on parameterizations of features in the field by their eigenvectors (medium eigenvectors for neutral surfaces and {\em dominant eigenvectors} for degenerate curves). These parameterizations lead to a more accurate extraction of degenerate curves and neutral surfaces than previous techniques based on the A-patches algorithm~\cite{Palacios:16}. However, while their extraction of degenerate curves is seamless, their extraction of neutral surfaces is not. In this paper, we provide a unified parameterization of all mode surfaces of a linear tensor field, which enables a unified pipeline for the seamless extraction of mode surfaces including neutral surfaces and degenerate curves.

Our analysis makes use results from 2D asymmetric tensor fields. Zheng and Pang~\cite{Zheng:05c} introduce the notion of {\em real domains} and {\em complex domains} for 2D asymmetric tensor fields. Zhang et al.~\cite{Zhang:09} provide topological analysis of asymmetric tensor fields on surfaces, which Khan et al.~\cite{Khan:20} extend to a multi-scale framework. Chen et al.~\cite{Chen:11} visualize asymmetric tensor fields on surfaces with a hybrid approach: glyphs for the complex domain and hyperstreamlines for the real domain.

\section{Tensor Background}
\label{sec:math_background}

In this section, we review the relevant math background on tensors and properties of 3D linear symmetric tensor fields.

\subsection{Tensor Basics}

An $n$-dimensional tensor $T$ can be expressed as an $n \times n$ matrix under a given orthonormal basis.

The {\em trace} of a tensor $T=(T_{ij})$ is the sum of its diagonal elements. When the trace is zero, the tensor is referred to as being {\em traceless}. A tensor $T$ can be uniquely decomposed as the sum of the tensor $D=\frac{\trace{T}}{n}\mathbb{I}$ (a multiple of the identity matrix) and a traceless tensor $A=T-D$ (referred to as the {\em deviator} of $T$). Note that $T$ and $A$ have the same set of eigenvectors. The set of all $n \times n$ tensors form a linear space, on which the following inner product of two tensors $R$ and $S$ can be introduced~\cite{spence2000elementary}:

\begin{equation}
  \langle R, S \rangle = \sum_{i=1}^{n}\sum_{j=1}^{n}R_{ij}S_{ij} = \trace(S^T R).
\end{equation}.

With this product, one can define the {\em magnitude} of a tensor $T$ as $||T||=\sqrt{\langle T, T \rangle}$. Another important quantity of a given tensor is its {\em determinant} $|T|$, which is the product of its eigenvalues.

A tensor $T$ is {\em symmetric} if it is equal to its transpose. Otherwise, it is {\em asymmetric}. The eigenvalues of a symmetric tensor are guaranteed to be real-valued, while the eigenvalues of an asymmetric tensor can be either real-valued or complex-valued. Furthermore, the eigenvectors belonging to different eigenvalues of a symmetric tensor form an orthonormal basis. For asymmetric tensors, even when the eigenvalues are real-valued, their respective eigenvectors are not mutually perpendicular.

Given our focus on 3D symmetric tensors and occasional mention of 2D asymmetric tensors, in the remainder of the paper we will drop the word ``symmetric'' for symmetric tensors and keep the word ``asymmetric'' for asymmetric tensors. Moreover, we only consider 3D traceless (symmetric) tensors and therefore will also omit the word ``traceless'' for 3D tensors. In contrast, when discussing 2D asymmetric tensors, we do not assume that they are traceless.

\subsection{3D Tensors and Modes}

A $3\times 3$ tensor $T$ has three eigenvalues $\lambda_1 \ge \lambda_2 \ge \lambda_3$, which are referred to respectively as its {\em major eigenvalue}, {\em medium eigenvalue}, and {\em minor eigenvalue}. Eigenvectors corresponding to $T$'s major eigenvalue are referred to its {\em major eigenvectors}. We can define $T$'s {\em medium eigenvectors} and {\em minor eigenvectors} in a similar fashion.

$T$ is {\em degenerate} if it has repeating eigenvalues. Under structurally stable conditions, a degenerate tensor $T$ has two eigenvalues being the same (referred to as the {\em repeating eigenvalue}). The third eigenvalue is the {\em dominant eigenvalue}. Furthermore, if the dominant eigenvalue is larger than the repeating eigenvalue, $T$ is referred to as being {\em linear degenerate}. If the dominant eigenvalue is smaller than the repeating eigenvalue, $T$ is referred to as being {\em planar degenerate}. The eigenvectors corresponding to the dominant eigennvalue are referred to as the {\em dominant eigenvectors}.

A $3 \times 3$ tensor $T$ is {\em neutral} if its medium eigenvalue is the average of its major and minor eigenvalues. The dominant eigenvalue and eigenvectors are {\em not} well-defined for neutral tensors.

The {\em mode} of a 3D (traceless, symmetric) tensor $T$ is $\mu(T) = 3\sqrt{6}\frac{\det(T)}{\|T\|^3}$, with a range of $[-1, 1]$. As special instances, neutral tensors are mode $0$ tensors, while linear degenerate tensors and planar degenerate tensors correspond to mode $1$ and mode $-1$ tensors, respectively. The eigenvalues of a tensor with a unit tensor magnitude can be expressed in terms of its mode $\mu$ as follows~\cite{irving2003integers,nickalls2006viete,CRISCIONE:00}:

\begin{gather}
	\lambda_1 = \sqrt{\frac{2}{3}} \sin (\frac{1}{3}\arcsin(-\mu) + \frac{2\pi}{3}), \nonumber \\
	\lambda_2 = \sqrt{\frac{2}{3}} \sin (\frac{1}{3}\arcsin(-\mu)), \nonumber \\
	\lambda_3 = \sqrt{\frac{2}{3}} \sin (\frac{1}{3}\arcsin(-\mu) - \frac{2\pi}{3}).    \label{eq:eigenvalue_and_mode}
\end{gather}

A 3D tensor field is a continuous tensor-valued function. A {\em degenerate point} and a {\em neutral point} are where the tensor values are {\em degenerate} and {\em neutral}, respectively. Under structurally stable conditions, degenerate points form curves ({\em degenerate curves})~\cite{Zheng:04}, and {\em neutral points} form surfaces ({\em neutral surfaces})~\cite{Palacios:16}. In general, the $\mu$ level set of the mode function is a surface when $-1<\mu<1$. Such a level set is referred to as a {\em mode}-$\mu$ {\em surface}~\cite{Palacios:16}. Note that both degenerate curves and neutral surfaces are special level sets of the mode $\mu$.

\subsection{3D Linear Tensor Fields}

We focus on 3D linear tensor fields, which can be written in the form of $T(x, y, z)=T_0+xT_x+yT_y+zT_z$ where $T_0$, $T_x$, $T_y$, and $T_z$ are linearly independent 3D tensors. Let $U$ be the set of 3D (traceless, symmetric) tensors, which is a five-dimensional space. Under structurally stable conditions, there exists a 3D tensor $\overline{T}$ that satisfies the following:

\begin{eqnarray}
\langle \overline{T}, T_0 \rangle=\langle \overline{T}, T_x \rangle=\langle \overline{T}, T_y \rangle=\langle \overline{T}, T_z \rangle=0, \nonumber \\
\langle \overline{T}, \overline{T} \rangle=1, \nonumber \\
\det(\overline{T}) \le 0.   \label{T_bar_condition}
\end{eqnarray}

Note that $\overline{T}$ plays an important role in the behavior of the tensor field~\cite{Roy:18}. We refer to $\overline{T}$ as the {\em characteristic tensor} of the linear tensor field.

The set of degenerate points of a 3D linear tensor field can be parameterized by a topological circle~\cite{Roy:18}. That is, the union of the set of mode $1$ points and mode $-1$ points is homeomorphic to the circle. The neutral surface of a 3D linear tensor field can be parameterized~\cite{Roy:18} by $\mathbb{RP}^2$ (referred to as the {\em medium eigenvector manifold}) except for two lines, each of which corresponds to a single point in the medium eigenvector manifold (referred to as a {\em singularity} in the parameterization). Due to the existence of the two singularities, the set of neutral points of a 3D linear tensor field is homeomorphic to $\mathbb{RP}^2$ attached with a handle (thus non-orientable)~\cite{armstrong1979basic}.

In this paper, we provide analysis on the topology of mode surfaces as well as efficient algorithms to extract them.

\section{Mode Surfaces Analysis}
\label{sec:mode_surface}

Given a number $\mu \in (-1, 0)\bigcup (0, 1)$, we wish to {\em seamlessly} extract the mode $\mu$ surface from a piecewise linear tensor field defined on a tetrahedral mesh. To do so, we provide a unified framework in the same spirit of the degenerate curve extraction method of Roy et al.~\cite{Roy:18}. That is, we first extract mode curves on each triangular face in the mesh. Next, we extend the mode curves from the four faces of each tet to extract the mode surface inside the tet. Finally, we stitch the mode surface from adjacent tets across their common faces to generate a seamless mode surface in the whole mesh.

This requires the ability to extract mode surfaces inside a tet at high-quality. Recall that inside each tet of the tetrahedral mesh, the tensor field is linear. In the remainder of this section, we will describe our novel analysis of mode surfaces for 3D linear tensor fields, which leads to a parameterization of such surfaces that enables high-quality extraction. We will state the results of our analysis in the paper and provide their proofs in Appendix~\ref{sec:proofs}.

Similar to the case of degenerate curves~\cite{Roy:18}, we will consider the set of mode $\pm \mu$ points together for our mode surface analysis and extraction. These points satisfy the following degree-six equation:

\begin{equation}
54(\det(T))^2 - \mu(T)^2\|T\|^6 =0, \label{eq:generalized_mode_formula}
\end{equation}

\noindent and we refer to the collection of such points as the {\em generalized mode $\mu$ surface}. As in the case of neutral surfaces, we show that a generalized mode $\mu$ surface can also be parameterized by its medium eigenvectors (Theorem~\ref{thm:medium_eigenvector_manifold} in Appendix~\ref{sec:proofs}). This parameterization is based on the following 2D asymmetric tensor field defined on the unit sphere:

\begin{equation}
 A(v_2) = R_{\frac{\theta}{2}+\frac{\pi}{4}} \overline{T}'(v_2) R_{\frac{\theta}{2}-\frac{\pi}{4}} \label{eq:asymmetry_field}
\end{equation}

\noindent in which $v_2$ is a unit vector, $\overline{T}'(v_2)$ is the projection of $\overline{T}$ onto the plane whose normal is $v_2$, $\theta = \arcsin(\sqrt{3}\tan(\frac{1}{3}\arcsin(\mu)))$, and $R_\phi=\begin{pmatrix}
            \cos\phi & -\sin\phi \\
            \sin\phi & \quad \cos\phi
          \end{pmatrix}$.

Figure~\ref{fig:mode_planes} illustrates the asymmetric tensor field $A$ with an example 3D linear tensor field created manually. There are four types of regions on the aforementioned sphere (the medium eigenvector manifold): (1) grey, (2) cyan, (3) magenta, and (4) blue.

\begin{figure}[tbh]
\centering%
\subfloat[][Medium eigenvector manifold]{
\includegraphics[height=1.65in]{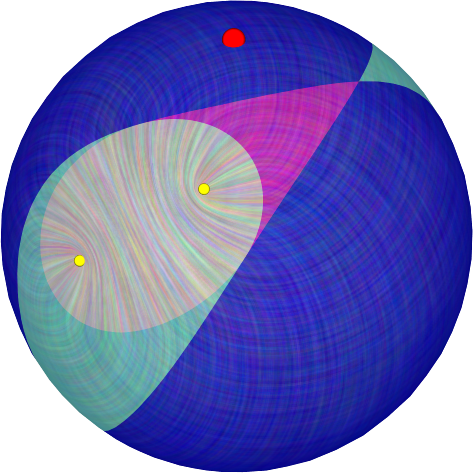}}
\hspace{0.05in}
\subfloat[][Generalized mode $\mu$ surfaces]{
\includegraphics[height=1.65in]{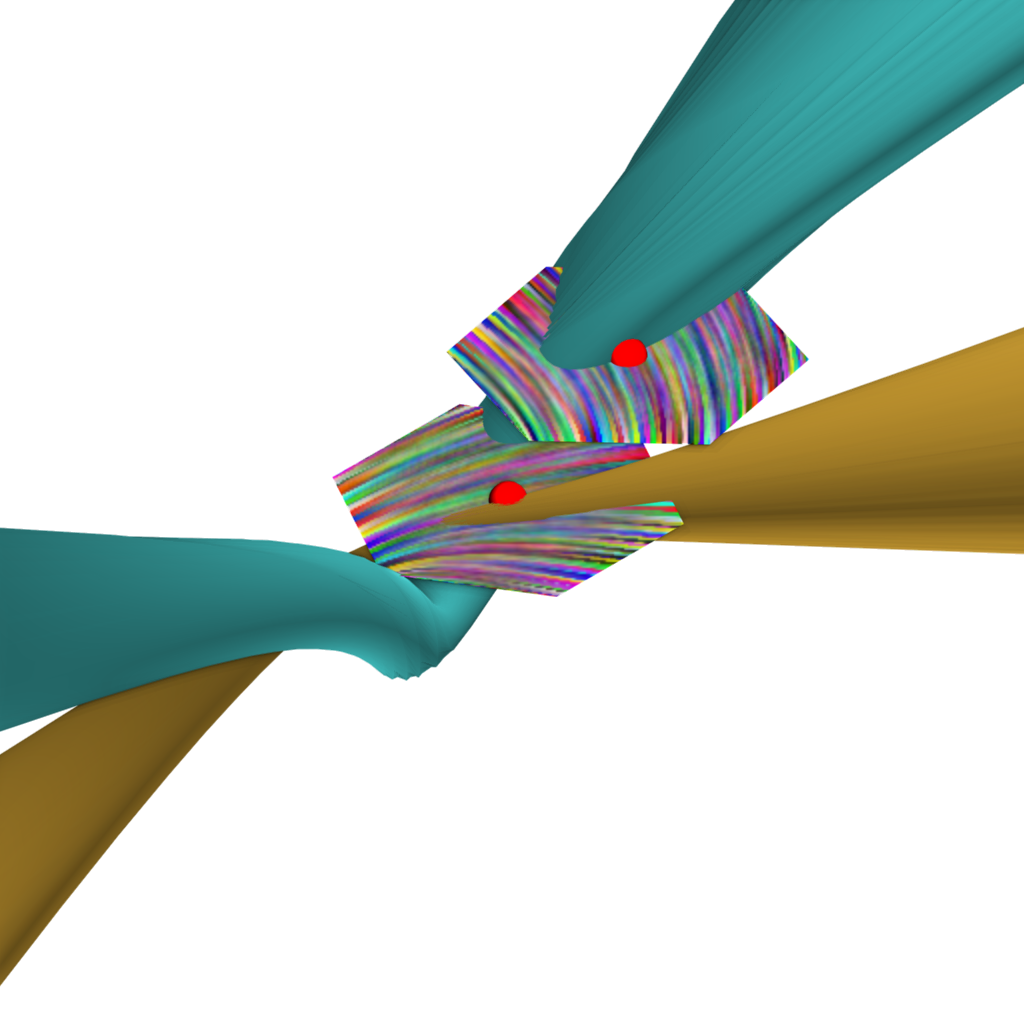}}\hspace{0.05in}
\caption[]{A generalized mode $\mu$ surface (right) can be parameterized by the medium eigenvector manifold (left). Each pair of antipodal points in the blue region of the medium eigenvector manifold (left: red dot) gives rise to two points in the generalized mode $\mu$ surfaces (right: red dots), one with a positive mode value (right: the red dot on the teal surface) and one with a negative mode value (right: the red dot on the gold surface). The dominant eigenvector directions at these points (right: shown with LIC textures in the planes) are given by the eigenvectors of an asymmetric tensor field in the medium eigenvector manifold (left: LIC texture directions). Points in the cyan and magenta regions in the medium eigenvector manifold correspond to two positive-mode points and two negative-mode points, respectively. Points in the grey region do not correspond to any point in the generalized mode $\mu$ surface.  }
\label{fig:mode_planes}
\end{figure}

The grey region is the {\em complex domain} of the asymmetric tensor field $A(v_2)$, i.e.\ with complex eigenvalues. Unit vectors in this region cannot appear as the medium eigenvector of any tensor in the 3D linear tensor field. That is, there are no points in the generalized mode $\mu$ surface that correspond to these unit vectors. Note that if a unit vector $v$ is in the complex domain, so is $-v$. Therefore, the complex domain of the asymmetric field respects the antipodal symmetry.

The cyan, magenta, and blue regions together form the {\em real domain} of the asymmetric tensor field. Each pair of antipodal points in the medium eigenvector manifold inside these regions correspond to two points in the generalized mode $\mu$ surface. If both points have the positive mode $\mu$, we color the original pair in the sphere with cyan. If both points have the mode $-\mu$, we color the pair in the sphere magenta. If one point has the mode $\mu$ and other $-\mu$, we color the corresponding pair in the sphere blue.

Note that the major and minor eigenvectors of the asymmetric tensor field $A$ give rise to the dominant eigenvectors of the 3D tensor field at the corresponding points in the generalized mode $\mu$ surface. An example is shown in Figure~\ref{fig:mode_planes}. In (a), a unit vector $v_2$ in the sphere (highlighted by a red dot) corresponds to two points in the generalized mode $\mu$ surface (b), each of which also highlighted by a red dot. The point in the teal surface (b) has a positive mode while the point in the gold surface has a negative mode. The planes normal to the medium eigenvectors are shown for both points (b). Notice that the two planes are parallel, i.e.\ the medium eigenvectors are the same for both points, which correspond to the vector in the medium eigenvector manifold (the red dot in (a)). The dominant eigenvector directions at the points in the generalized mode $\mu$ surface ((b): the LIC textures in the two planes) together match the major and minor eigenvector directions of $A(v_2)$ ((a): the LIC directions at the red dot).

\begin{figure*}[tbh]
    \centering%
    $\begin{array}{@{\hspace{0.0in}}c}
    \includegraphics[width=6.6in]{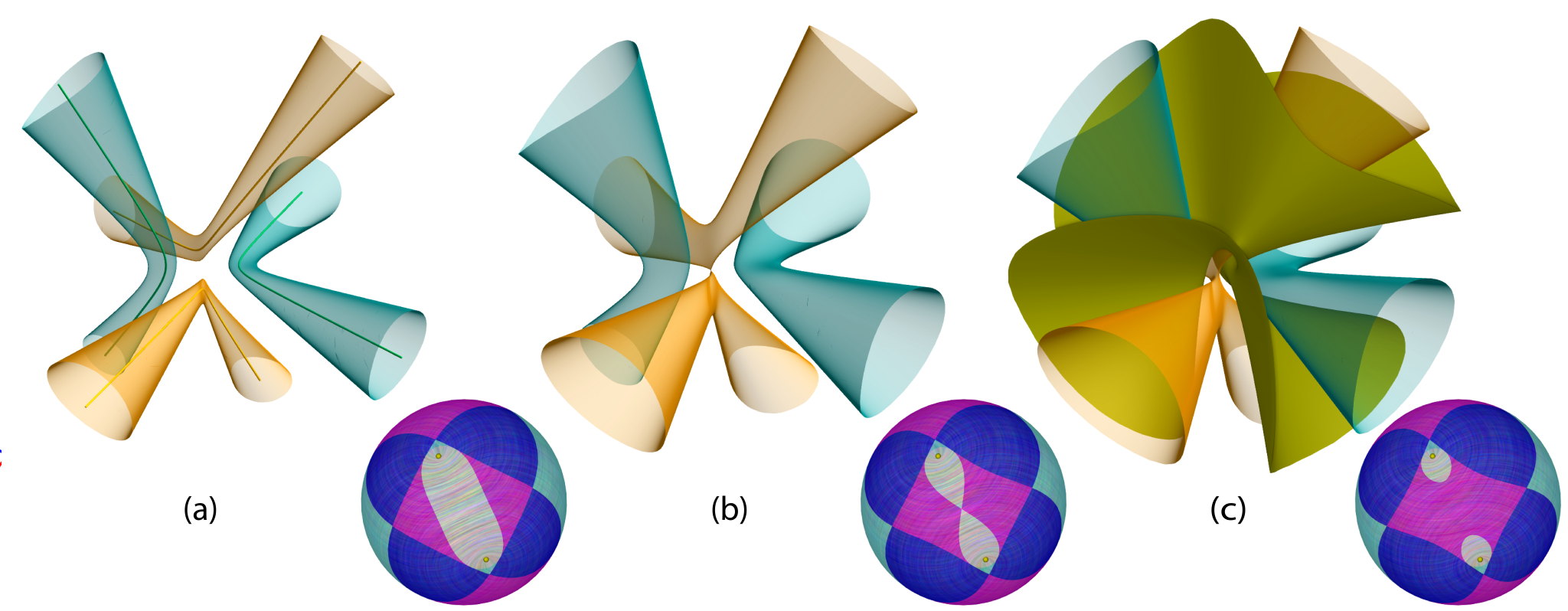}
    \end{array}$
\caption[]{Given a 3D linear tensor field, the positive mode surfaces (teal) converge to linear degenerate curves (green) and the negative mode surfaces (gold) converge to planar degenerate curves (yellow) when $\mu$ approaches $\pm 1$ (a). In this case, the complex domain consists of two loops respecting the antipodal symmetry of the sphere. Together, the generalized mode $\mu$ surface has a topology of a torus. In contrast, when $\mu$ approaches $0$ (c), the positive mode and negative mode surfaces together converge onto the neutral surface (chartreuse). The complex domain now consists of two pairs of antipodal loops, and the generalized mode $\mu$ surface has a topology of a double-torus. The bifurcation between the two cases occurs (b) when the complex domain resembles a figure-eight and the generalized mode $\mu$ surface has a non-manifold point.    }
\label{fig:topologymodesurface}
\end{figure*}

Due to the symmetry in the tensor field, $-v_2$ corresponds to the same two points in the generalized mode $\mu$ surface (Theorem~\ref{thm:medium_eigenvector_manifold} in Appendix~\ref{sec:proofs}). However, the minor eigenvector of $A(v_2)$ becomes the major eigenvector of $A(-v_2)$ and the major eigenvector of $A(v_2)$ becomes the minor eigenvector of $A(-v_2)$. To make our parameterization easier for subsequent processing, we choose the unit vector from $v_2$ and $-v_2$ so that its major eigenvector gives rise to the dominant eigenvector of the corresponding point in the generalized mode $\mu$ surface. This scheme removes the ambiguity in our parameterization by converting the two-to-two correspondence ($\pm v_2$ to the two points in the generalized mode $\mu$ surface) to a one-to-one correspondence.

Finally, points on the boundary between the real domain and complex domain ({\em complex domain boundary}) have one real eigenvalue of multiplicity of two. The complex domain boundary also satisfies the antipodal symmetry. That is, if a vector $v_2$ is in the complex domain boundary, so is $-v_2$. Moreover, $v_2$ and $-v_2$ correspond to exactly one point in the generalized mode $\mu$ surface.

Given a 3D linear tensor field, all of its generalized mode $\mu$ surface can be parameterized by the same sphere. The following equation characterizes the complex domain boundary (Theorem~\ref{thm:complex_boundary_equation} in Appendix~\ref{sec:proofs}):

\begin{equation}
\frac{1}{2} - v_2^T \overline{T}^2 v_2 + \frac{1}{4} \cos^2\theta(v_2^T \overline{T} v_2)^2 = 0 \label{eq:real_boundary_0}
\end{equation}

\noindent where $\theta = \arcsin(\sqrt{3}\tan(\frac{1}{3}\arcsin(\mu)))$ is the same as that in Equation~\ref{eq:asymmetry_field}. We illustrate the changes in the geometry and topology of generalized mode $\mu$ surfaces with an example tensor field in Figure~\ref{fig:topologymodesurface}. Three generalized mode $\mu$ surfaces (teal and gold surfaces) are shown in (a), (b), and (c), respectively. Their corresponding eigenvector manifolds are also shown (the three spheres). In addition, we show the degenerate curves (the yellow and green curves in (a)) and the neutral surfaces (the chartreuse surface in (c)).

When $\mu=1$, the mode surface is essentially the degenerate curves. The complex domain in the corresponding medium eigenvector manifold (not shown) consists of two connected components, satisfying the antipodal symmetry. When $\mu$ decreases, the complex domain shrinks in size (a). At this stage, the complex domain boundary still consists of two curves satisfying the antipodal symmetry (one is visible in (a)). Since each pair of antipodal points on the complex domain boundary corresponds to exactly one point in the generalized mode $\mu$ surface, the latter can be topologically constructed by removing the complex domains from the sphere (punching two holes) and gluing the surface along the complex domain boundary based on the antipodal symmetry. This results in a torus.

Note that as $\mu$ decreases, the complex domain grows smaller (Corollary~\ref{cor:domain_decreasing} in Appendix~\ref{sec:proofs}). When $\mu$ reaches $\sqrt{1-\overline{\mu}^2}$ where $\overline{\mu}$ is the mode of $\overline{T}$, the complex domain boundary touches itself (Figure~\ref{fig:topologymodesurface} (b)). The generalized mode $\mu$ surface at this point is a non-manifold.

After this, the complex domain continues to shrink and now has four connected components. Performing the same topological surgery (removing the complex domains and gluing along the complex domain boundary) results in a double-torus, i.e.\ a surface with two handles (Figure~\ref{fig:topologymodesurface} (c)).

Finally, when $\mu=0$, the complex domain disappears, and the complex domain boundary degenerates into two pairs of antipodal points in the medium eigenvector manifold. They are precisely the singularities in the parameterization~\cite{Roy:18}, each of which corresponds to a straight line in the neutral surface. Using the eigenvectors of $\overline{T}$ as the coordinate system for the medium eigenvector manifold, we have the singularities being $\begin{pmatrix} \pm \sqrt{\frac{\overline{\lambda}_1-\overline{\lambda}_2}{\overline{\lambda}_1-\overline{\lambda}_3}} & 0 & \pm \sqrt{\frac{\overline{\lambda}_2-\overline{\lambda}_3}{\overline{\lambda}_1-\overline{\lambda}_3}} \end{pmatrix}$.

We refer the reader to Theorem~\ref{thm:bifucation} in Appendix~\ref{sec:proofs} for the proof of the above analysis. Note that neutral surfaces are the only non-orientable mode surfaces, which highlights their topological significance.

The complex domain boundary is characterized by a degree-six polynomial, which makes its extraction challenging. Fortunately, similar to degenerate curves in a 3D linear tensor field which can be characterized by an elliptical loop, the complex domain boundary can also be parameterized as illustrated in Figure~\ref{fig:complex_para}. Due to the four-fold symmetry in the complex domain boundary, it is sufficient to focus on one quarter of the curve, which, in the coordinate system
	\begin{align}
		\alpha &=  v_2^T \overline{T} v_2, \\
		\beta  &=  v_2^T \overline{T}^2 v_2
	\end{align}
	is characterized by
	\begin{align}
		0 &= 1 - 2 \beta + \frac{1}{2} \cos^2\theta \alpha^2.
		\label{eqn:real_complex_bound_alpha_mainpaper}
	\end{align}

This is the equation of a parabola, which can be parameterized by $\alpha$ (Lemma~\ref{lemma:complex_boundary_parameterizable} in Appendix~\ref{sec:proofs}). Each $\alpha$ gives one corresponding $\beta$ from Equation \ref{eqn:real_complex_bound_alpha_mainpaper} and thus one point in each of the four quarters of the complex domain boundary (Figure~\ref{fig:complex_para}).

\begin{figure}[!htb]
\begin{center}
\includegraphics[width = 2.0in]{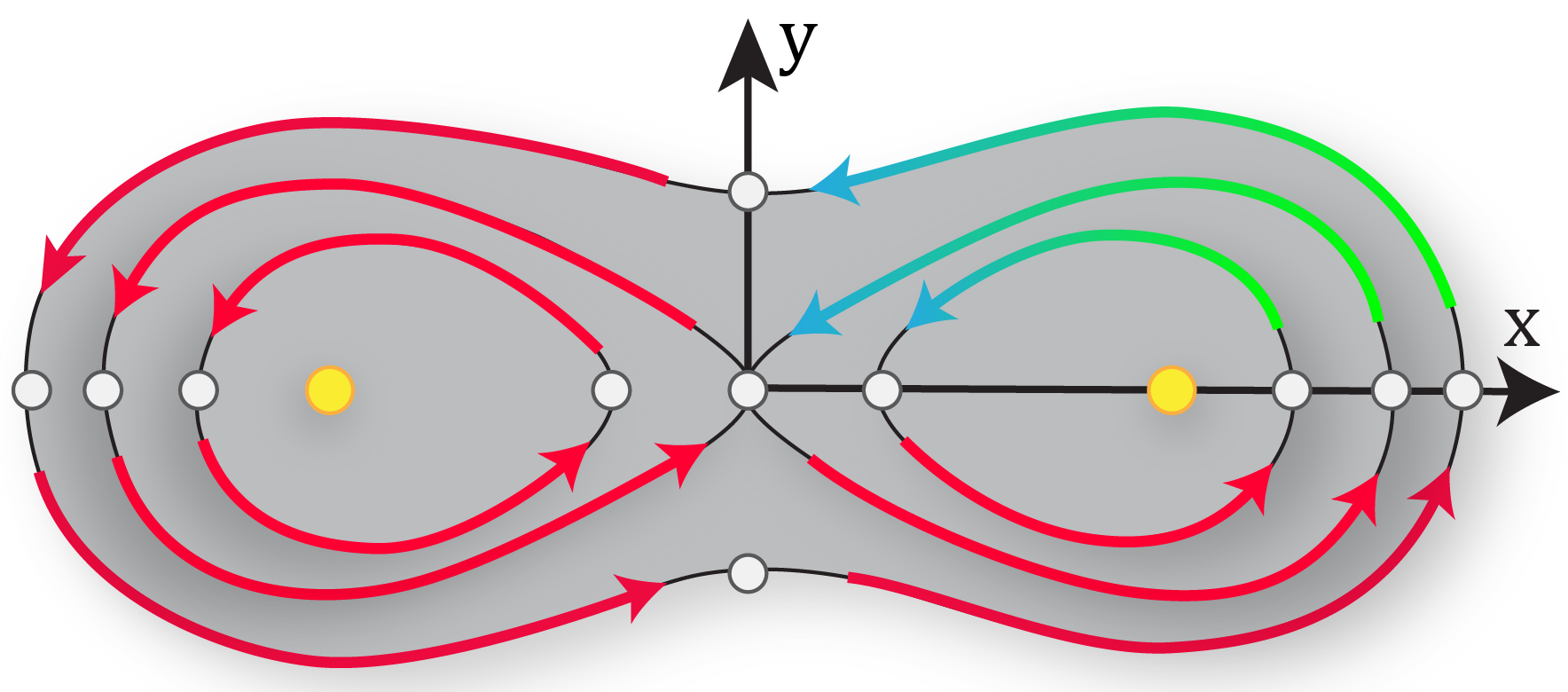}
\end{center}
\caption{The complex domain boundary has a four-fold symmetry, with each quarter parameterizable. A quarter segment can intersect the horizontal axis and the vertical axis at one point each, making the complex domain a single region (the outermost loop). This corresponds to the case shown in Figure~\ref{fig:topologymodesurface} (a). In contrast, a quarter segment can also intersect the horizontal axis only (inner-most case), splitting the complex domain into two regions (the innermost loop). This corresponds to the case shown in Figure~\ref{fig:topologymodesurface} (c). The bifurcation occurs when the quarter segments pass through the origin (the loop between the outermost and innermost loops), which corresponds to Figure~\ref{fig:topologymodesurface} (b) .
}
\label{fig:complex_para}
\end{figure}

\section{Extraction of Mode Surfaces}
\label{sec:extraction}

In this section, we describe our mode surface extraction algorithm, the input of which is a tetrahedral mesh, at whose vertices tensors are given. These tensor values are linearly interpolated into the faces and interiors of each tet. This results in a piecewise linear 3D tensor field, i.e.\ inside each tet, the tensor field is linear. Similarly, inside each face, we have a 2D linear tensor field.

\begin{figure}[!htb]
\begin{center}
   \includegraphics[width=2.6in]{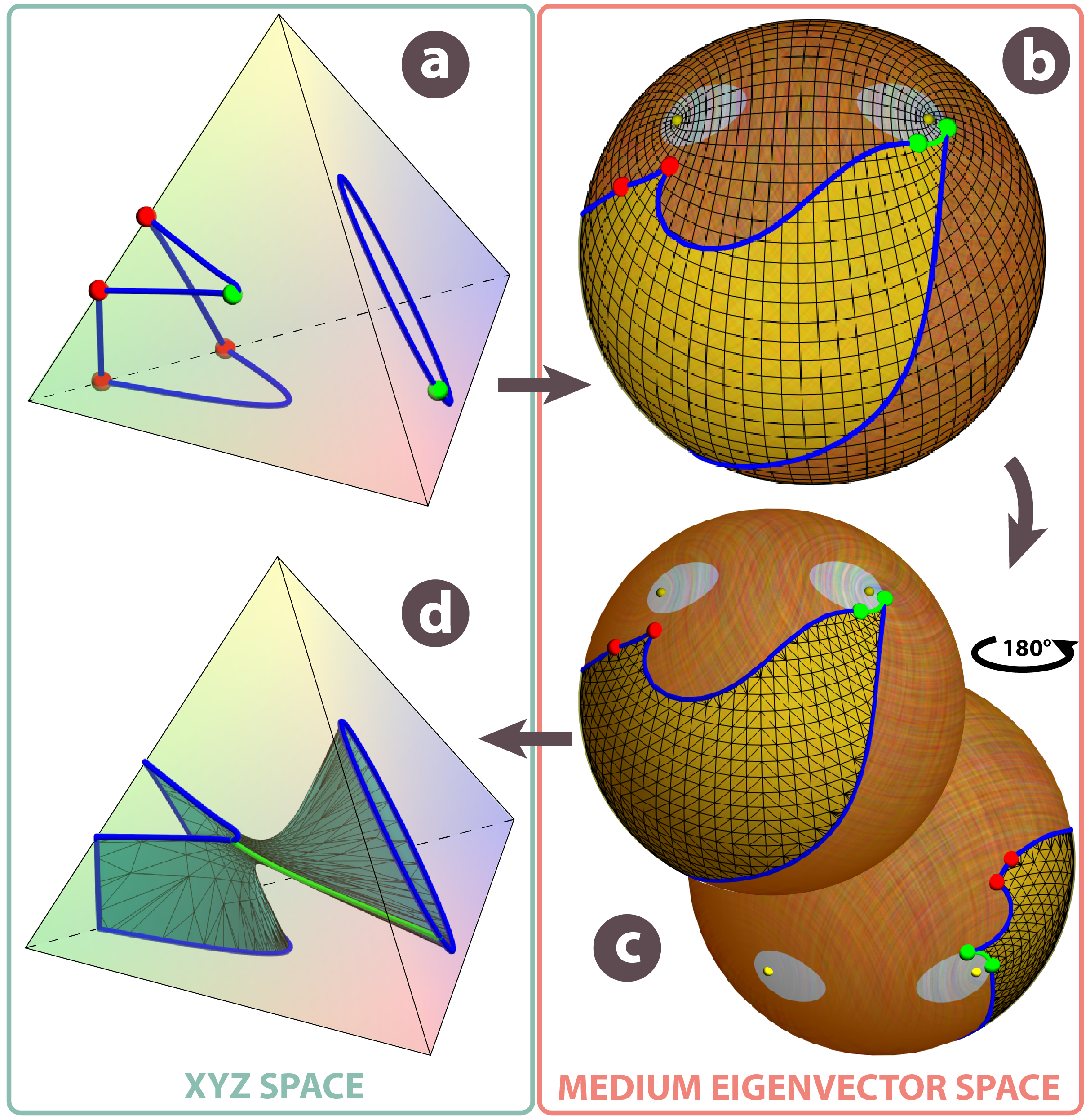}
\caption{The pipeline of our mode extraction algorithm for a given tet: starting from mode curves extracted from the faces of a tet (a), we identify the region in the medium eigenvector manifold bounded by the aforementioned curves (b). The region is then triangulated (c) and mapped back to the $XYZ$ space to give the mode surface in the tet (d).  }\label{fig:mode_extraction}
\end{center}
\end{figure}

The pipeline of our unified mode surface (including neutral surfaces) extraction is illustrated in Figure~\ref{fig:mode_extraction} with an example tensor field. First, we extract mode curves inside each face in the mesh, including internal loops. Next, for each tet with known mode curves in any of its faces, we extract the mode surface inside the tet. Finally, mode surfaces extracted from the tets will be stitched together along shared faces. We describe the detail of each step next.

\subsection{Mode Curve Extraction inside a Face}

Mode curves inside a plane can consist of a number of open curves and loops. If an open curve does not intersect any edge of a triangle in our mesh, the curve must be entirely outside the triangle and is not part of mode curves for the triangle. In contrast, there exist loops that are entirely inside the triangle (thus valid) but do not intersect any of the edges of the triangle. Thus, detecting mode curves for a triangle by only detecting their intersections with the boundary edges can miss the inner loops.

To overcome this problem, we use a property from Morse theory~\cite{matsumoto2002introduction}, which states that for any level set loop of a function there must be at least one local maximum or minimum of that function enclosed by the loop. Based on this insight, we first extract the critical points of the mode function inside the plane containing the triangle. Next, for each extremum (a critical point that is not a saddle) inside the triangle, we insert a new vertex at the point and subdivide the triangle into three triangles. In Theorem~\ref{thm:at_most_four_extrema} (Appendix~\ref{sec:proofs}) we show that there are at most four extrema inside a plane $P$ containing the triangle. Only the ones that are inside the triangle are valid. Thus, the aforementioned subdivision process is executed at most four times for each face in the mesh. To find the critical points of the mode function inside a plane $P$, we make use of the fact that they must be solutions to the following system of polynomial equations:

\begin{eqnarray}
vh(u, v, w)=wg(u, v, w),  \label{eq:critical_ax_main} \\
wf(u, v, w)=uh(u, v, w),  \label{eq:critical_ay_main}\\
u^2 + v^2 + w^2 = 1
\end{eqnarray}

\noindent where $\nabla \det(u T_1 + v T_2 + w T_3)=(f(u,v,w), g(u,v,w), h(u,v,w))$. Here, $T_1$, $T_2$, and $T_3$ is an orthonormal basis for the set of tensor values inside $P$ given by the tensor field. In addition, $(u, v, w)$ is a unit vector such that $\frac{T(x,y)}{||T(x,y)||}=u(x,y) T_1 + v(x,y) T_2 + w(x,y) T_3$. Note that Equations~\ref{eq:critical_ax_main} and~\ref{eq:critical_ay_main} are both homogeneous cubic polynomials. Moreover, if $(u,v,w)$ is a solution, so is $(-u, -v, -w)$. Furthermore, both $(u,v,w)$ and $(-u, -v, -w)$ correspond to the same tensor, i.e.\ the same critical point. Let $u'=\frac{u}{w}$ and $v'=\frac{v}{w}$. The original system of equations is transformed into the following system of two cubic equations:

\begin{eqnarray}
v'\overline{h}(u', v')=\overline{g}(u', v'), \nonumber \\
\overline{f}(u', v')=u'\overline{h}(u', v')
\end{eqnarray}

\noindent where $\overline{f}$, $\overline{g}$, and $\overline{h}$ are non-homogeneous cubic polynomials derived respectively from $f$, $g$, and $h$ with the change from $u, v, w$ to $u', v'$. According to B\'ezout's theorem~\cite{Fulton:74}, this system of cubic equations is equivalent to a degree-nine polynomial which has nine solutions. However, we show that two of the solutions are spurious and four other solutions correspond to degenerate points in the plane $P$. Moreover, the spurious solutions can be found by solving $h(u, v, 0)=0$ and $u^2+v^2=1$, while the degenerate points can be found using the method of Roy et al.~\cite{Roy:18}. Once these solutions are factored out, we obtain a cubic polynomial whose roots, when real-valued, are additional critical points (local extrema and saddles). We compute these roots using the Eigen library~\cite{eigenweb}. For each solution in the form of $u$, $v$, and $w$ values, we recover $T=u T_1 + v T_2 + w T_3$ which we use to find $(x, y)$ such that $T(x, y)=T$.

We can extract the intersection points of the mode curves with the edges in the triangles (including the subdivided triangles). Starting from these intersection points, we perform numerical tracing in the direction perpendicular to the gradient of the mode function. This guarantees that all internal loops in the mode curves are found. To trace out the face intersections we use a numerical ODE integrator. This leads to mode curves whose ends are on the edges of the face. Tracing is finished when we exit the (possibly subdivided) triangle. We then run a numerical root-finding algorithm~\cite{Alefeld1995} on the interpolation function produced by the ODE integrator to accurately find the point where it exits, which is connected to the closest edge intersection point.

\subsection{Mode Surface Extraction inside a Tetrahedron}

Once we have extracted mode curves from all the faces in the mesh, we proceed to extract mode surfaces from inside each tet.

First, we gather the mode curves extracted from the four faces of the tet and stitch open curves (those intersecting the edges of the tet) into loops (Figure~\ref{fig:mode_extraction} (a): blue curves).

Next, for each sample point on the aforementioned mode curves, we find the corresponding point on the medium eigenvector manifold for the 3D linear tensor field in the tet. Therefore, the mode curves (represented as polylines) on the boundary of the tet lead to a set of curves (also represented as polylines) in the medium eigenvector manifold (Figure~\ref{fig:mode_extraction} (b): blue curves).

Recall that a pair of antipodal points on the complex domain boundary corresponds to the same point in the mode surface. Therefore, a connected component in the polyline in the tet's boundary can become disconnected in the medium eigenvector manifold at the complex domain boundary (green points in Figure~\ref{fig:mode_extraction} (a) and (b)). To overcome this problem, we need to not only find the exact intersection points of the polylines with the complex domain boundary but also generate curves connecting such points in the medium eigenvector manifold, which correspond to segments of the complex domain boundary (Figure~\ref{fig:mode_extraction} (b): green curves).

Detecting the existence of such points is relatively straightforward as each such point corresponds to an antipodal pair in the medium eigenvector manifold. Therefore, along a mode curve, such a point divides the curve into two disconnected components, whose medium eigenvectors are situated on opposite hemispheres of the medium eigenvector manifold. Consequently, for every segment in the polyline, we check the dot product between the medium eigenvectors of the two end points. If the sign of the dot product is negative, we mark the segment as intersecting the complex domain boundary and proceed to find the exact location of the intersection point through numerical root-finding using Equation~\ref{eq:real_boundary_0}. This segment is then split into two with the insertion of the newly found intersection point.

Once we have found all the complex domain boundary points in the faces of the tet, we compute their corresponding antipodal point pairs on the medium eigenvector manifold. To decide which segments on the complex domain boundary are part of the mode surface, we make use of the parameterization for the complex domain boundary (Figure~\ref{fig:complex_para}). This parameterization allows us to select a new point on the complex domain boundary and check whether its corresponding point on the mode surface is inside or outside the tet. This information is then used to decide which segments along the complex domain boundary are part of the mode surfaces inside the tet. Notice that our method is similar to the method of Roy et al.~\cite{Roy:18} for computing degenerate curves, which are parameterizable by an ellipse.

The segments along the complex domain boundary (Figure~\ref{fig:mode_extraction} (b): green curves) and the polylines representing the mode curves from the faces (Figure~\ref{fig:mode_extraction} (b): blue curves) bound the region in the medium eigenvector manifold that corresponds to the mode surface inside the tet (Figure~\ref{fig:mode_extraction} (b): yellow region). To find the interior of this region, we perform a constrained Delaunay triangulation on the set of sample points (from mode curves and complex domain boundary segments). Since Delaunay triangulation is normally defined on a plane, we take the stereographic projection of the vertices as the stereographic projection maps circles to circles and thus preserves the Delaunay condition. This leads to a triangulation of the region with rather poor aspect ratios for the triangles.

To improve the quality of the tessellation, we add more points inside the region on the medium eigenvector manifold. Given our unified approach of using the medium eigenvector manifold for mode surface extraction (including neutral surfaces), we make use of the following quad parameterization of the medium eigenvector manifold based on the singularities in the neutral surfaces.

\begin{align}
	x &= \cos\theta\sqrt{f\cos^2\eta + 1},  \nonumber \\
	y &= \cos\theta \sin\eta, \nonumber \\
	z &= \cos\eta,
\end{align}

\noindent where $f = \frac{2\lambda_2(\overline{T})-\lambda_3(\overline{T}) -\lambda_1(\overline{T})}{\lambda_3(\overline{T})-\lambda_2(\overline{T})}$, $0 \le \eta < \pi$, and $0 \le \theta < 2 \pi$. This parameterization leads to a perfect quadrangulation of the medium eigenvector manifold (Figure~\ref{fig:mode_extraction} (b): the quad grid) with the only irregular vertices in the quadrangulation being the singularities in the medium eigenvector manifold (Figure~\ref{fig:mode_extraction} (b): yellow dots). After the grid points inside the region are added, we perform a second constrained Delaunay triangulation with both boundary and interior sample points. This leads to improved aspect ratios in the triangulation (Figure~\ref{fig:mode_extraction} (c)).

Finally, we need to map the sample points from the medium eigenvector manifold back to the $XYZ$ space while preserving the connectivity among them to construct the mode surfaces inside the tetrahedron. To map a point $v_2$ in the medium eigenvector manifold to its corresponding point in the $XYZ$ space, we need to identify the tensor $T$ whose medium eigenvector is $v_2$. The eigenvalues of $T$ can be computed using~Equation~\ref{eq:eigenvalue_and_mode} given the mode value $\mu$. The dominant eigenvector of $T$ can be computed from the asymmetric tensor $A(v_2)$ (Equation~\ref{eq:asymmetry_field}). This gives us the tensor $T$. We can find its corresponding point in the $XYZ$ space by solving a system of linear equations~\cite{Roy:18}.

\subsection{Stitching Mode Surfaces from Adjacent Tets}

To get a seamless mode surface, we need to stitch the sheets from different tets across their common faces.

This step is relatively straightforward as we have computed mode curves in the faces first. Therefore, mode surfaces from neighboring tets already have matching polylines with the same vertices.

\subsection{Neutral Surface Extraction}

We use the same pipeline for both mode surfaces and neutral surfaces, with the following two differences.

First, neutral surfaces and curves are characterized by degree-three polynomials. Therefore, when finding the intersection with an edge, we can directly compute them using the cubic formula.

Second, since there is no complex domain in the medium eigenvector manifold for neutral surfaces, we do not need to detect the intersection of mode curves with the complex domain boundary. Instead, we need to find the singularities in the medium eigenvector manifold, which correspond to straight lines in the $XYZ$ space. Their intersection with the boundary faces can be easily computed. Then the intersection points corresponding to the same singularity are connected using a line segment.

Other than these differences, our framework handles mode surfaces and neutral surfaces in a unified fashion.

\section{Performance}
\label{sec:performance}

Our methods can extract mode surfaces and neutral surfaces with higher quality, i.e.\ the extracted surfaces are seamless and more accurate than existing mode surface and neutral surface extraction methods~\cite{Palacios:16,Roy:18} (Figure~\ref{fig:quality_compare}). The seamlessness of the extracted surfaces enables additional information about mode surfaces to be computed more robustly, such as principal curvature directions (Figure~\ref{fig:quality_compare_IBFV}).

In addition, our technique is faster than existing techniques. On average, our neutral surface extraction method is about $1.8$ times faster than the hybrid method introduced by Roy et al.~\cite{Roy:18}, and our mode surface extraction method is $5.3$ times faster than the A-patches method in~\cite{Palacios:16}. 
Measurements were taken on a computer with Intel(R) Xeon(R) E$3$-$2124$G CPU$@$ $3.40$ GHz, $64$GB of RAM, and an NVIDIA Quadro P$620$ GPU. We used four test data sets: a block with a single compression force (Figure~\ref{fig:teaser}: $480000$ tets), a block with a compression force and an extension force (Figure~\ref{fig:07blockSince}: $450000$ tets), a block with three compression forces (Figure~\ref{fig:26_frontAndBackView}: $384000$ tets), and a block with a twisting compression force (Figure~\ref{fig:24_frontAndBackView} in Appendix~\ref{sec:additional_example}: $286416$  tets).

%

\section{Applications}
\label{sec:apps}

The application of compressive loads is ubiquitous in engineering where solids are extruded to designed geometries and liquids are compressed in combustion engines, and in medical research where human or animal organs are compressed to enable successful imaging processes. On the other hand, compression is a challenging state to model and elucidate.
In fluid mechanics, Navier-Stokes modeling~\cite{chorin1990mathematical} applies to incompressible fluids and needs a completely different solution scheme when the fluids become compressible. In solid mechanics, a scalar quantity referred to as the Poisson’s ratio~\cite{greaves2011poisson} is used to dial from polymers, which hardly compress, to materials that change volume during compression. Here we examine four scenarios of compression in a solid block by visualizing the mode surfaces of the stress tensor fields. We vary boundary conditions on this block to reveal the material behavior embedded in the stress tensors. Three scenarios (Figure~\ref{fig:24_AbaqusU2}) are discussed in this section while the fourth is covered in Appendix~\ref{sec:additional_example}.

\begin{figure}[!b]
\centering%
\subfloat[][]{\includegraphics[height=0.5in]{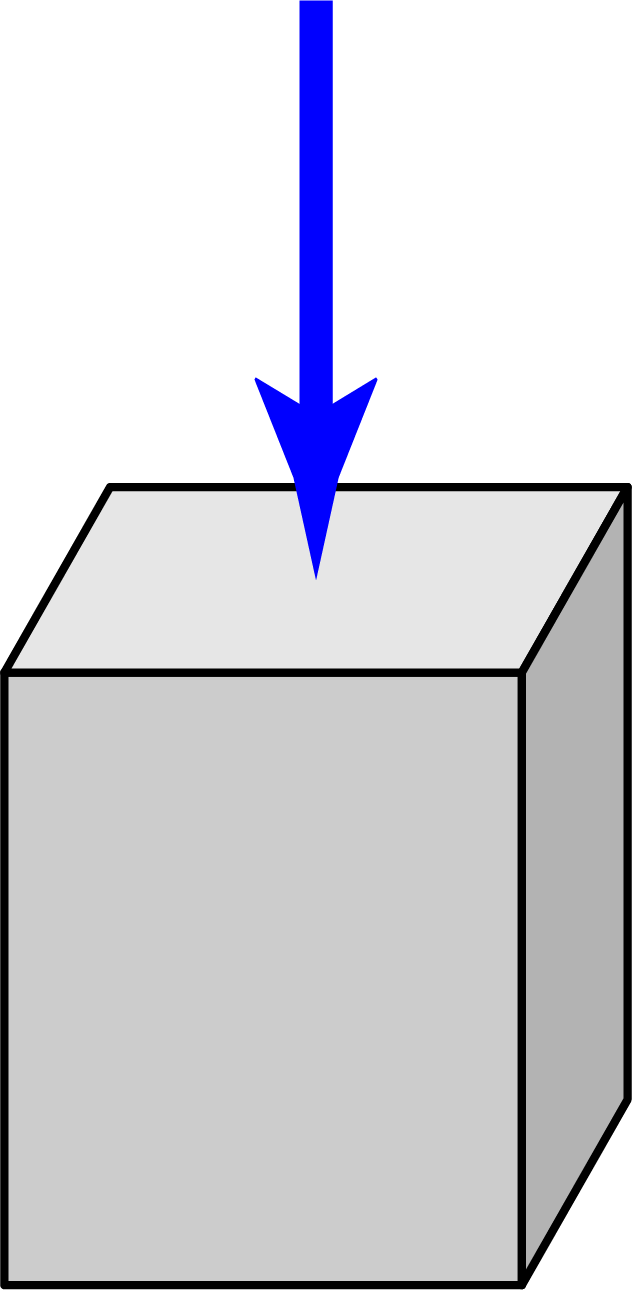}\label{fig:geometry1}}
\hspace{0.2in}
\subfloat[][]{\includegraphics[height=0.5in]{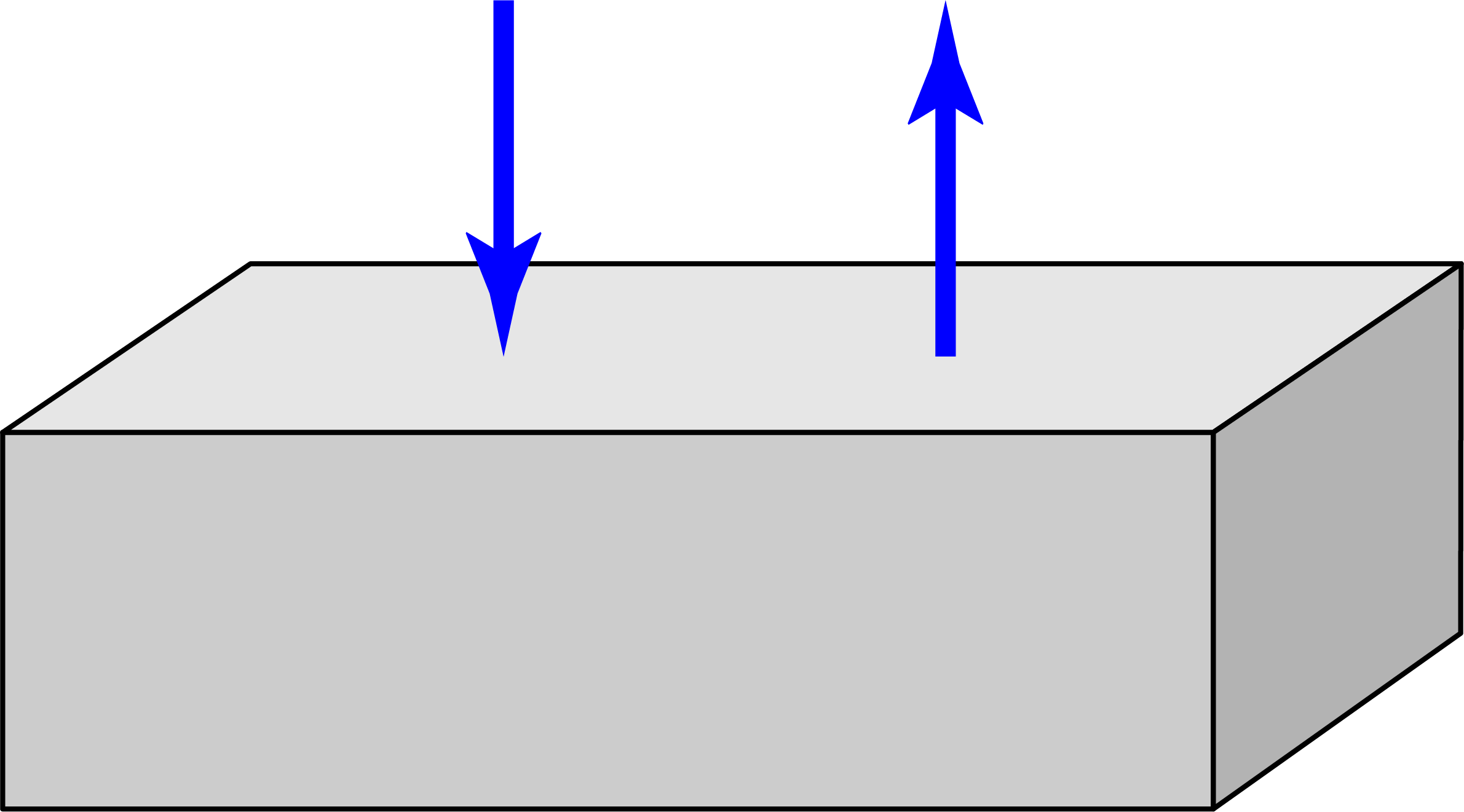}\label{fig:geometry1}}
\hspace{0.05in}
\subfloat[][]{\includegraphics[height=0.5in]{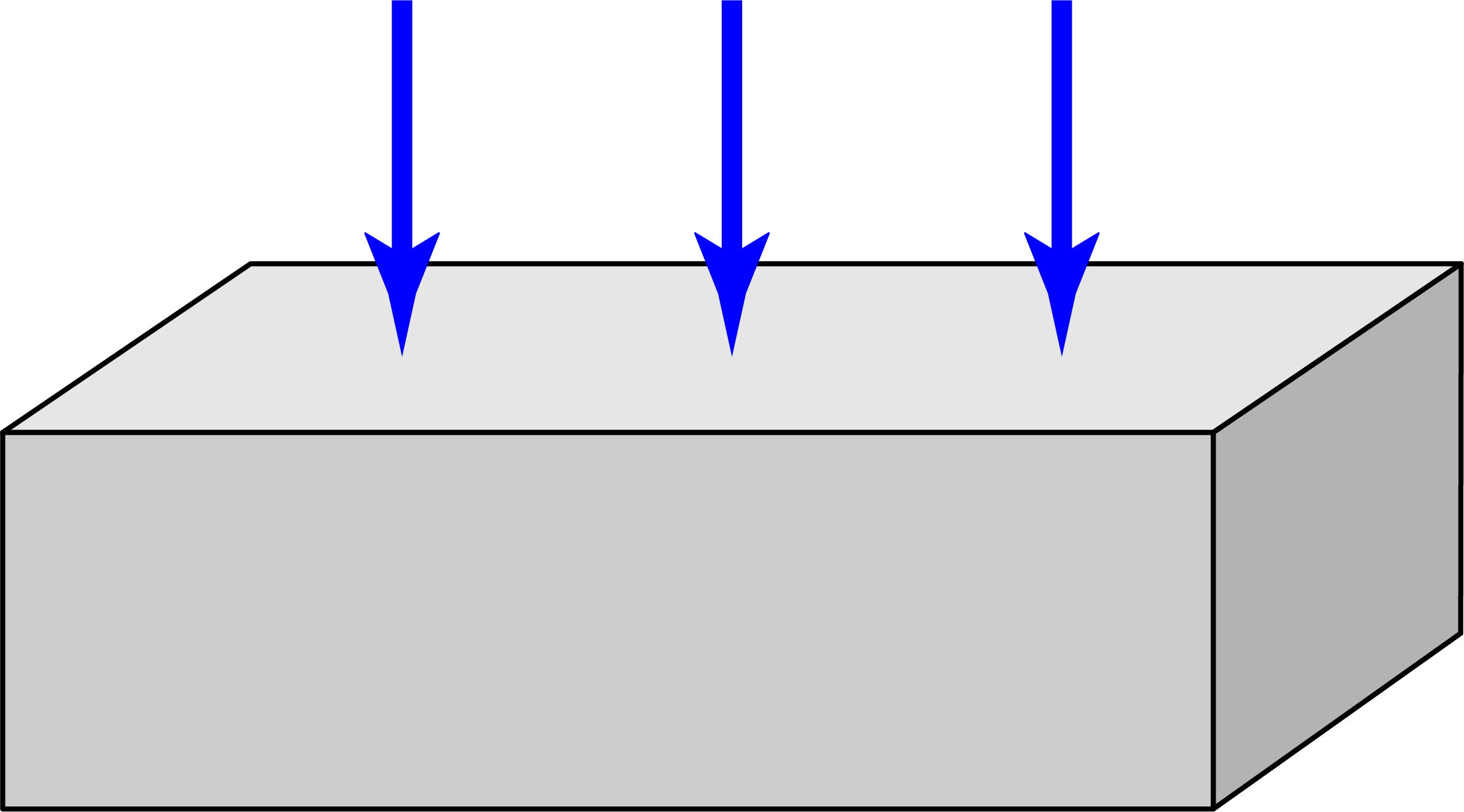}\label{fig:geometry3}}
\caption[]{Three scenarios of compression in a solid block.}
\label{fig:24_AbaqusU2}
\end{figure}

In the first scenario (Figure~\ref{fig:24_AbaqusU2} (a)), there is a compression force on the top of a cubic block. The resulting mode surfaces are displayed in Figure~\ref{fig:teaser}. Notice that the volume is dominated by negative mode surfaces (gold: indicating compression-dominant). In addition, the network of degenerate curves leads to an interesting vascular structure that is not commonly observed. This indicates that despite a simple boundary condition, the shape of the block can also play an important role in deciding where uniaxial compression can occur.

We contrast compression against extension in our second scenario (Figure~\ref{fig:24_AbaqusU2} (b)), in which one side of the block is pushed down and the other side is pulled up forming a dome-shaped dent and a dome, respectively. Our visualization (Figure~\ref{fig:07blockSince}) confirms this, as all planar degenerate curves (yellow: uniaxial compression) appear in the left half and all linear degenerate curves (green: uniaxial extension) appear in the right half of the block. In addition, we show the generalized mode $0.97$ surface. Notice the symmetry between the mode surfaces for compression (left side) and extension (right side).  While compression and extension are usually conceptualized as two different types of deformations (respectively volume loss and gain), to our knowledge the symmetry between them in the generalized mode surfaces is new to the research and application communities. In addition, Figure~\ref{fig:quality_compare_IBFV} (b) shows the major principal curvature directions for the generalized mode $0.99$ surface. Notice that in the negative mode part of this surface (gold), the major principal curvature direction is mostly aligned with the width of the block except in the middle where it is aligned with the length of the block. The sudden change in the directions is a reflection of the fact that due to the boundary condition, the middle part of this surface is being pushed out more along the block than across the block. The collision of compression directions leads to some trisectors on the surface near the middle. Due to symmetry, similar deformations of the surface (this time extension) can be observed on the positive mode part of the surface (teal). The number and location of the umbilical points (e.g. trisectors) can indicate uniform compression or extension and will be investigated further. Note that such observations are impossible without the ability to generate seamless meshes for mode surface (e.g. compare this to Figure~\ref{fig:quality_compare_IBFV} (a)).

\begin{figure}[tb]
\centering%
\includegraphics[width=2.8in]{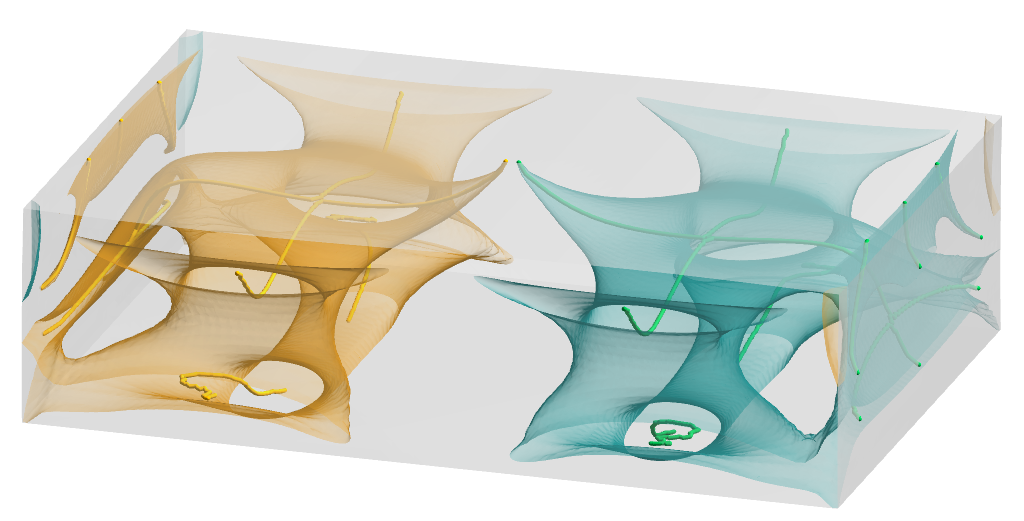}
\caption[]{Mode surfaces of stress tensor fields for a rectangle block being pushed and pulled simultaneously. The left side is compressed while the right side is extended through boundary conditions applied at the top.}
\label{fig:07blockSince}
\end{figure}

Our third example studies spatially consecutive compressive loads (Figure~\ref{fig:24_AbaqusU2} (c)) which can be seen on highway bridges where fleets of heavy trucks are parked due to traffic lights or in structures where pairs of nuts and bolts are installed at a number of evenly spaced locations. Intuitively, we can see that these compressive ``holding-down'' points create a wave form in the underlying geometry. In our example where we have three consecutive compressive loads, we note that through analyzing the mode surfaces of the stress tensor field (Figure~\ref{fig:26_frontAndBackView} (a)), this wave shape can be observed at mode value $-0.53$ (gold).  As the mode value moves towards zero, the wave shape subsides (Figure~\ref{fig:26_frontAndBackView} (b)). After this mode value, only the periodic behavior of the regions between the compressive loads remains. Between each pair of compressive loads, there is a tubular region indicating that the material between the loads can be in both extension and compression at the same time.

\begin{figure}[!tb]
  \centering
\subfloat[][$\mu=\pm 0.53$]{\includegraphics[width=2.8in]
{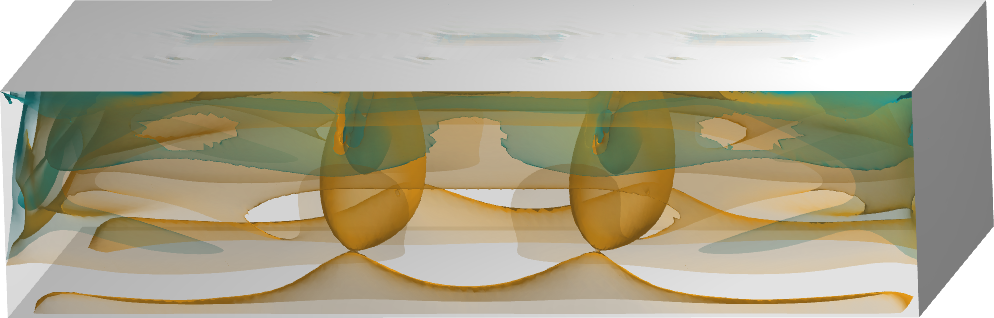}}\\
\subfloat[][$\mu=\pm 0.50$]{\includegraphics[width=2.8in]
{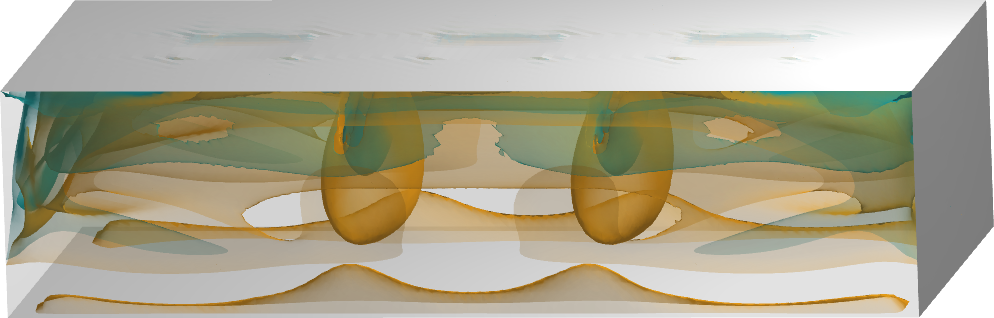}}
\caption{Mode surfaces of stress tensor fields of a rectangular block being pushed down at three locations.    }\label{fig:26_frontAndBackView}
\end{figure}

\section{Conclusion and Future Work}
\label{sec:conclusion}

In this paper, we provide to our knowledge the first approach to the seamless extraction of mode surfaces (including neutral surfaces). The method is faster and of higher quality than existing methods of mode (neutral) surface extraction. Together with the degenerate curve extraction algorithm of Roy et.~\cite{Roy:18}, our approach represents the first unified approach for the seamless extraction of mode surfaces of any mode. At the core of the approach is our novel topological analysis of mode surfaces (including degenerate curves and neutral surfaces) using the same medium eigenvector manifold.

In addition, we apply our technique to a number of data sets to focus on analyzing stress tensor fields for compressive behavior that is fundamental in solid mechanics.  We present a new demonstration of compression in volumes to indicate different behavior from different regions of the material.

Our current sampling strategy for the medium eigenvector manifold does not guarantee best quality of the triangle mesh such as the aspect ratios of the triangles. We plan to explore other sampling patterns to improve on this.

There are other types of feature surfaces in 3D symmetric tensor fields, such as extremal surfaces~\cite{Zobel2017} as well as magnitude surfaces and anisotropy index surfaces~\cite{Palacios:16}. Adapting our parameterization approach to those surfaces is also a future research avenue.

Extending tensor field analysis and visualization to more complex physical behaviors that are hybrids of different kinds of deformation can also be fruitful future research directions. Furthermore, correlating patterns observed in the geometry of mode surfaces to the load configurations has the potential of increasing our understanding and thus control in applications such as tire performance and bridge maintenance. We plan to explore these avenues in our future research.

\acknowledgments{
We wish to thank our anonymous reviewers for their valuable suggestions. We thank Kyle Hiebel for making the voice recording of our video. This research is partially supported by NSF awards (\# 1566236) and (\# 1619383).}

\bibliographystyle{abbrv-doi}

\bibliography{ms}

\appendix
\normalsize
\clearpage

\section{Theoretical Results and Proofs}
\label{sec:proofs}
\setcounter{theorem}{0}

\begin{theorem}
Given a 3D linear tensor field $T(x, y, z)=T_0+xT_x+yT_y+zT_z$, a mode value $0 < \mu < 1$, and a unit vector $v_2$, the number of points on the generalized mode $\mu$ surface with $\pm v_2$ as its medium eigenvector is the same as the number of real eigenvalues of the 2D asymmetric tensor $A=R_{\frac{\theta}{2}+\frac{\pi}{4}}\overline{T}'R_{\frac{\theta}{2}-\frac{\pi}{4}}$ where $\overline{T}'$ is the projection of characteristic tensor $\overline{T}$ onto the plane $P$ with normal $v_2$, $\theta = \arcsin(\sqrt{3}\tan(\frac{1}{3}\arcsin(\mu)))$, and $R_{\phi}=\begin{pmatrix}\cos\phi & -\sin\phi \\ \sin\phi & \quad \cos\phi\end{pmatrix}$. The real-valued eigenvectors of $A$ give rise to the dominant eigenvectors of the corresponding points in the generalized mode $\mu$ surface.
\label{thm:medium_eigenvector_manifold}
\end{theorem}

\begin{proof}

As pointed out in~\cite{Roy:18}, given a tensor $t=\lambda_1 v_1 v_1^T + \lambda_2 v_2 v_2^T + \lambda_3 v_3 v_3^T$ in the linear tensor field $T(x, y, z)=T_0+xT_x+yT_y+zT_z$, where $v_1$, $v_2$, and $v_3$ are respectively the major, medium, and minor eigenvectors, we have

\begin{equation}
v_1^T \overline{T} v_1 + v_2^T \overline{T} v_2 + v_3^T \overline{T} v_3 = 0. \label{eq:T_bar_trace}
\end{equation}

In addition, since $\overline{T}$ is the characteristic tensor of $T(x, y, z)$,

\begin{align}
	0 &= \left\langle \overline{T}, \lambda_1 v_1 v_1^T + \lambda_2 v_2 v_2^T + \lambda_3 v_3 v_3^T \right\rangle \nonumber\\
	&= \lambda_1 \trace(\overline{T} v_1 v_1^T) + \lambda_2 \trace(\overline{T} v_2 v_2^T) + \lambda_3 \trace(\overline{T} v_3 v_3^T) \nonumber\\
	&= \lambda_1 v_1^T \overline{T} v_1 + \lambda_2 v_2^T \overline{T} v_2 + \lambda_3 v_3^T \overline{T} v_3.
\end{align}

\noindent according to the {\em cyclic property of trace}~\cite{Greene:14}.

Combining this equation with Equation~\ref{eq:T_bar_trace}, we have
\begin{align}
	0 &= (\lambda_1 - \lambda_2) v_1^T \overline{T} v_1 - (\lambda_2 - \lambda_3) v_3^T \overline{T} v_3 .
	\label{eqn:v1v2v3_in_field}
\end{align}

Notice that the major eigenvector $v_1$ and minor eigenvector $v_3$ must be inside $P$, the plane that contains the point where $t$ occurs in the field and whose normal is $v_2$. Let $v_1'$ and $v_3'$ be $v_1$ and $v_3$ expressed in the coordinate system of $P$. Consequently, Equation~\ref{eqn:v1v2v3_in_field} can be rewritten as the following:
\begin{align}
	0 &= (\lambda_1 - \lambda_2) v_1'^T \overline{T}' v_1' - (\lambda_2 - \lambda_3) v_3'^T\overline{T}' v_3'.
	\label{eqn:v1v2v3_prime_in_field}
\end{align}

We first consider the case when $\mu(t)>0$ and use the right-handed coordinate system where $v_3$ is the horizontal axis and $v_1$ is the vertical axis.

For simplification purposes, we define $u = k v_1' + l v_3'$ and $w = k v_1' - l v_3'$ where $k = \sqrt{\frac{\lambda_1-\lambda_2}{\lambda_1-\lambda_3}}$ and $l = \sqrt{\frac{\lambda_2-\lambda_3}{\lambda_1-\lambda_3}}$. Therefore, $v_1' = \frac{u + w}{2 k}$ and $v_3' = \frac{u - w}{2 l}$.

It is straightforward to verify that $u$ and $w$ both have unit length.
Moreover, since both $v_1'$ and $v_3'$ are unit vectors and have the same length, it can be verified that

\begin{equation}
u\cdot w = k^2 - l^2 = \frac{\lambda_1 +\lambda_3 -2\lambda_2}{\lambda_1 - \lambda_3}.
\end{equation}

\noindent Since $\lambda_1+\lambda_2+\lambda_3=0$, we have

\begin{equation}
\lambda_1 +\lambda_3 -2\lambda_2 = -3\lambda_2  = \sqrt{6} \sin (\frac{1}{3}\arcsin(\mu)).
\end{equation}

\noindent Similarly, it can be shown that

\begin{equation}
\lambda_1-\lambda_3=\sqrt{2-3\lambda_2^2}=\sqrt{2}\cos(\frac{1}{3}\arcsin(\mu)).
\end{equation}
\noindent Consequently,
\begin{equation}
u \cdot w= \frac{\lambda_1 +\lambda_3 -2\lambda_2}{\lambda_1 - \lambda_3} =  \sqrt{3}\tan(\frac{1}{3}\arcsin(\mu)) = \sin\theta.
\end{equation}

Notice that since $k \ge 0$ and $l \ge 0$, $w$ is always counterclockwise of $u$. Therefore, $w$ is obtained by rotating $u$ counterclockwise by $\frac{\pi}{2} - \theta$. Note that Equation~\ref{eqn:v1v2v3_prime_in_field} can be rewritten in terms of $u$ and $w$ as $w^T \overline{T}' u = 0$, or equivalently $w \cdot (\overline{T}' u) = 0$. That is, $w \bot \overline{T}' u$.

As $v_1'$ is a bisector of $u$ and $w$, $v_1'$ can be obtained by either rotating $u$ counterclockwise by $\frac{\pi}{4}-\frac{\theta}{2}$ or rotating $w$ counterclockwise by $\frac{\theta}{2}-\frac{\pi}{4}$. That is, $u=R_{\frac{\theta}{2}-\frac{\pi}{4}}v_1'$ and $w=R_{\frac{\pi}{4}-\frac{\theta}{2}} v_1'$. Therefore,
\begin{equation}
R_{\frac{\pi}{4}-\frac{\theta}{2}} v_1' \perp \overline{T}' R_{\frac{\theta}{2}-\frac{\pi}{4}}v_1'.
\end{equation}

Since $v_3' \perp v_1'$, we have $v_1' = R_{-\frac{\pi}{2}}v_3'$. Consequently, $R_{\frac{\pi}{4}-\frac{\theta}{2}} R_{-\frac{\pi}{2}} v_3' \perp \overline{T}' R_{\frac{\theta}{2}-\frac{\pi}{4}}v_1'$, which is equivalent to $R_{-\frac{\pi}{4}-\frac{\theta}{2}} v_3' \perp \overline{T}' R_{\frac{\theta}{2}-\frac{\pi}{4}}v_1'$. Rotating both sides conterclockwise by $\frac{\pi}{4}+\frac{\theta}{2}$, we have

\begin{equation}
v_3' \perp R_{\frac{\pi}{4}+\frac{\theta}{2}} \overline{T}' R_{\frac{\theta}{2}-\frac{\pi}{4}}v_1'.
\end{equation}

Recall that $A=R_{\frac{\theta}{2}+\frac{\pi}{4}}\overline{T}'R_{\frac{\theta}{2}-\frac{\pi}{4}}$. Then $v_3' \perp A v_1'$.

This means that $A v_1'$ must be orthogonal to a vector that is orthogonal to $v_1'$. As $v_1'$ is in two dimensions, this implies that $A v_1'$ is a scalar multiple of $v_1'$. Consequently, $v_1'$ is an eigenvector of $A$.

In the case where $\mu(t)<0$, we use the right-handed coordinate system of $P$ such that $v_1$ is now the horizontal axis and $v_3$ is the vertical axis.

We again define $u = k v_1' + l v_3'$ but negate $w$, i.e.\ $w = -k v_1' + l v_3'$. We can still show that $u$ and $w$ are unit vectors and $u\cdot w=\sqrt{3}\tan(\frac{1}{3}\arcsin(\mu))$.

From here, we follow the same argument for the case where $\mu(t)>0$ except that we need to reverse the roles of $v_1'$ and $v_3'$. This leads to that $v_3'$ is an eigenvector of $A=R_{\frac{\theta}{2}+\frac{\pi}{4}}\overline{T}'R_{\frac{\theta}{2}-\frac{\pi}{4}}$.

We now consider a number of cases. First, when $A$ has zero real eigenvalues, its eigenvectors must be complex-valued. Since the eigenvectors of a 3D symmetric tensor cannot be complex-valued, there are no points on the generalized mode $\mu$ surface with $v_2$ as its medium eigenvectors.

When $A$ has only one real eigenvalue, there is only one point in the generalized mode $\mu$ surface whose medium eigenvector is $v_2$. Depending on the sign of the mode value of this point, either its major eigenvector $v_1$ ($\mu(t)>0$) or minor eigenvector $v_3$ ($\mu(t)<0$) is given by the corresponding eigenvector of $A$. Recall that the dominant eigenvector is the major eigenvector $v_1$ when $\mu(t)=\mu>0$ and the minor eigenvector $v_3$ when $\mu(t)=-\mu<0$. Therefore, the eigenvector of $A$ gives the dominant eigenvector of the point in the generalized mode $\mu$ surface.

To see the uniqueness, we note that if there is another point in the generalized mode $\mu$ surface with the same combination of $v_1$, $v_2$, $v_3$ and $\mu(t)$, then the tensor must be a multiple of $t$. However, as pointed out by Roy et al.~\cite{Roy:18}, if a tensor $t$ appears in a 3D linear tensor field, then none of its multiples can appear in the same field. Consequently, there is only one point in the generalized mode $\mu$ surface with $v_2$ as its medium eigenvector.

On the other hand, when $A$ has two real eigenvalues, the major eigenvector and the minor eigenvector of $A$ {\em each} corresponds to a point in the generalized mode $\mu$ surface whose medium eigenvectors are given by $v_2$. If the point has a positive mode value, its major eigenvector $v_1$ is given by the corresponding eigenvector of $A$. In contrast, if the point has a negative mode value, its minor eigenvector $v_3$ is given by the corresponding eigenvector of $A$.

What remains to be shown is that using $-v_2$ as the medium eigenvector, we arrive at the same set of points as using $v_2$. Again, we have the two cases $\mu(t)>0$ and $\mu(t)<0$.

For the first case, note that the rotations in the definition of $A$ are around $v_2$, and they are reversed when $-v_2$ is chosen as the medium eigenvector. Assume that $A$'s eigenvalues are $\hat{\lambda}_1$ and $\hat{\lambda}_2$, and that $v_1'$ is the eigenvector corresponding to $\hat{\lambda}_1$. Then we find the eigenvalue on the opposite side of the sphere with

\begin{align}
	R_{-(\frac{\theta}{2}+\frac{\pi}{4})}\overline{T}'R_{-(\frac{\theta}{2}-\frac{\pi}{4})} v_1'
	&= A^T v_1' \nonumber\\
	&= R_{\frac{\pi}{2}} (R_{-\frac{\pi}{2}} A^T R_{\frac{\pi}{2}}) R_{-\frac{\pi}{2}} v_1'  \nonumber\\
	&= R_{\frac{\pi}{2}} adj(A) R_{-\frac{\pi}{2}} v_1' \nonumber\\
	&= R_{\frac{\pi}{2}} adj(A) v_3' \nonumber\\
	&= R_{\frac{\pi}{2}}  \hat{\lambda}_2 v_3' \nonumber\\
	&= \hat{\lambda}_2 R_{\frac{\pi}{2}} v_3' \nonumber\\
	&= \hat{\lambda}_2 v_1'
\end{align}
where we have used a property of asymmetric $2$x$2$ tensors~\cite{spence2000elementary} that if $A = \begin{bmatrix} a & b \\ c & d \end{bmatrix}$ is a 2D asymmetric tensor and $R_\theta$ is the two-dimensional counterclockwise rotation matrix of angle $\theta$,
\begin{align}
	R_{-\frac{\pi}{2}} A^T R_{\frac{\pi}{2}} = \begin{bmatrix} \quad d & -b \\ -c & \quad a \end{bmatrix} = \adj(A)
	\label{eqn:rotation_adjugate}
\end{align}
is the adjugate of $A$, whose eigenvectors are the same as $A$'s and the roles of whose major and minor eigenvalues are swapped. Consequently, the point on the generalized mode $\mu$ surface with $v_1$ given by the major eigenvector of $A(v_2)$ is the same point whose $v_1$ is given by the minor eigenvector of $A(-v_2)$.

A similar argument applies when $\mu(t)<0$.

To summarize, $v_2$ and $-v_2$ correspond to the same set of points in the generalized mode $\mu$ surface. Moreover, when there are two points with $v_2$ and $-v_2$ as medium eigenvectors, each of $A(v_2)$ and $A(-v_2)$ corresponds to exactly {\em one} point in the generalized mode $\mu$ surface whose dominant eigenvector (as a 3D symmetric tensor) is given by the major eigenvector of $A$ (as an asymmetric tensor).
\end{proof}

\begin{theorem}
Given a 3D linear tensor field $T(x, y, z)=T_0+xT_x+yT_y+zT_z$ and a mode value $\mu$, the real domain in the medium eigenvector manifold corresponding to $\mu$ is characterized by $\frac{1}{2} - v_2^T \overline{T}'^2 v_2 + \frac{1}{4} \cos^2\theta(v_2^T \overline{T}' v_2)^2 \ge 0$ where $\theta = \arcsin(\sqrt{3}\tan(\frac{1}{3}\arcsin(\mu)))$. The boundary between the real and complex domain occurs if and only if the equal sign holds. \label{thm:complex_boundary_equation}
\end{theorem}

\begin{proof}
Let $\overline{T'}$ be the projection of the characteristic tensor $\overline{T}$ of the tensor field onto the plane perpendicular to unit medium eigenvectors $v_2$. We consider the asymmetric tensor field $A(v_2)=R_{\frac{\theta}{2}+\frac{\pi}{4}}\overline{T}'(v_2) R_{\frac{\theta}{2}-\frac{\pi}{4}}$.

As mentioned in~\cite{Zhang:09}, a $2 \times 2$ asymmetric matrix $A$ can be uniquely decomposed as follows:

\begin{equation}
	A = \gamma_d \begin{bmatrix} 1 & 0 \\ 0 & 1 \end{bmatrix}
	+ \gamma_r \begin{bmatrix} 0 & -1 \\ 1 & \quad  0 \end{bmatrix}
	+ \gamma_s \begin{bmatrix} \cos \tau & \quad \sin \tau \\ \sin \tau & -\cos \tau \end{bmatrix}
\end{equation}

\noindent where $\gamma_d$, $\gamma_r$, and $\gamma_r$ are the isotropic, rotational, and anisotropic components, respectively. Note that $\tau$ gives rise to the eigenvector information of $A$. Depending on the discriminant of $\Delta=\gamma_s^2-\gamma_r^2$, $A$ has either two real eigenvalues ($\Delta>0$) or two complex-valued eigenvalues ($\Delta<0$). When $\Delta=0$, $A$ has a pair of repeating eigenvalues. For our asymmetric tensor field, this implies that $v_2$ is on the complex domain boundary in the medium eigenvector manifold. Notice that this condition is both necessary and sufficient. To compute $\Delta$, we need to compute both $\gamma_d$ and $\gamma_r$.

	\begin{align}
		\gamma_d
		&= \frac{1}{2} \trace(A) \nonumber\\
        &= \frac{1}{2} \trace(R_{\frac{\theta}{2}+\frac{\pi}{4}}\overline{T}'R_{\frac{\theta}{2}-\frac{\pi}{4}}) \nonumber\\
        &= \frac{1}{2} \trace(R_{\frac{\theta}{2}-\frac{\pi}{4}}R_{\frac{\theta}{2}+\frac{\pi}{4}}\overline{T}')    \nonumber   \\
		&= \frac{1}{2} \trace(R_\theta \overline{T}') \nonumber\\
		&= \frac{1}{2} \trace(\cos\theta \,\mathbb{I}\, \overline{T}') + \frac{1}{2} \trace\left(\sin\theta \begin{bmatrix} 0 & -1 \\ 1 & \quad 0 \end{bmatrix} \overline{T}'\right) \nonumber\\
		&= \cos\theta \frac{\trace(\overline{T}')}{2}
		\label{eqn:scale_A_trace}
	\end{align}
	where the second term vanishes because $\overline{T'}$ is symmetric.

Since trace is the same in any basis, we have $$\trace(\overline{T}') = v_1'^T \overline{T}' v_1' + v_3'^T \overline{T}' v_3'.$$ Recall that $\overline{T}'$, $v_1'$, and $v_3'$ are respectively the projection of $\overline{T}$, $v_1$, and $v_3$ onto the plane perpendicular to $v_2$. Thus,

\begin{equation}
v_1'^T \overline{T}' v_1' + v_3'^T \overline{T}' v_3' = v_1^T \overline{T} v_1 + v_3^T \overline{T} v_3 = -v_2^T \overline{T}v_2
\end{equation}

\noindent in which the second equality above comes from rearranging \autoref{eq:T_bar_trace}.

\noindent Consequently,

\begin{equation}
\gamma_d = -\cos\theta\frac{v_2^T \overline{T}v_2}{2}.
\end{equation}

A similar calculation shows that

\begin{equation}
  \gamma_r = \sin\theta \frac{\trace(\overline{T}')}{2} = -\sin\theta \frac{v_2^T \overline{T}v_2}{2}.
\end{equation}

Notice that $\gamma_d^2+\gamma_r^2=(\frac{v_2^T \overline{T} v_2}{2})^2$. To compute $\gamma_s$ it is convenient to find the squared magnitude of $A$, which is

\begin{equation}\trace(A^T A) = \trace((R_{\frac{\theta}{2}+\frac{\pi}{4}}\overline{T}'R_{\frac{\theta}{2}-\frac{\pi}{4}})^T R_{\frac{\theta}{2}+\frac{\pi}{4}}\overline{T}'R_{\frac{\theta}{2}-\frac{\pi}{4}}) = \trace(\overline{T}'^2).
\end{equation}

Since $v_2$ has unit length, the matrix $\mathbb{I} - v_2 v_2^T$ is a projection matrix and is thus equal to its square. Using this to project the tensor $\overline{T}$, we have
	\begin{align}
		\trace(\overline{T}'^2)
		&= \trace((\mathbb{I} - v_2 v_2^T) \overline{T} (\mathbb{I} - v_2 v_2^T)^2 \overline{T} (\mathbb{I} - v_2 v_2^T)) \nonumber\\
		&= \trace( \overline{T} (\mathbb{I} - v_2 v_2^T)^2 \overline{T} (\mathbb{I} - v_2 v_2^T)^2) \nonumber\\
		&= \trace(\overline{T} (\mathbb{I} - v_2 v_2^T) \overline{T} (\mathbb{I} - v_2 v_2^T)) \nonumber\\
		&= \trace(\overline{T}^2) - \trace(\overline{T} v_2 v_2^T\overline{T}) \nonumber\\
		&- \trace(\overline{T}^2 v_2 v_2^T) + \trace(\overline{T}v_2 v_2^T\overline{T}v_2 v_2^T) \nonumber\\
		&= 1 - 2 v_2^T \overline{T}^2 v_2 + (v_2^T \overline{T} v_2)^2.
		\label{eqn:magnitude_A}
	\end{align}
	Because $A$'s magnitude is $2 (\gamma_r^2 + \gamma_s^2 + \gamma_d^2)$, we can use this to find $\gamma_s$.
	\begin{align}
		\gamma_s^2
		&= \frac{1}{2} \trace(A^T A) - \gamma_d^2 - \gamma_r^2 \nonumber\\
		&= \frac{1}{2} - v_2^T \overline{T}^2 v_2 + \frac{1}{2} (v_2^T \overline{T} v_2)^2 - \left( \frac{v_2^T \overline{T} v_2}{2} \right)^2 \nonumber\\
		&= \frac{1}{2} - v_2^T \overline{T}^2 v_2 + \frac{1}{4} (v_2^T \overline{T} v_2)^2.
	\end{align}

The real domain is characterized by $\gamma_s^2 \ge \gamma_r^2$, i.e.\

\begin{equation}
\frac{1}{2} - v_2^T \overline{T}^2 v_2 + \frac{1}{4} (v_2^T \overline{T} v_2)^2 \ge (-\sin\theta \frac{v_2^T \overline{T}v_2}{2})^2
\end{equation}

or equivalently, $\frac{1}{2} - v_2^T \overline{T}^2 v_2 + \frac{1}{4} \cos^2\theta(v_2^T \overline{T} v_2)^2 \ge 0$. The complex domain boundary is thus

\begin{equation}
\frac{1}{2} - v_2^T \overline{T}^2 v_2 + \frac{1}{4} \cos^2\theta(v_2^T \overline{T} v_2)^2 = 0 \label{eq:real_boundary}
\end{equation}

\end{proof}

\begin{corollary}
Given a 3D linear tensor field $T(x, y, z)=T_0+xT_x+yT_y+zT_z$ and two mode values $\mu_1$ and $\mu_2$ such that $0\le \mu_1 < \mu_2\le 1$, the complex domain on the medium eigenvector manifold corresponding to $\mu_1$ is a proper subset of that corresponding to $\mu_2$. \label{cor:domain_decreasing}
\end{corollary}

\begin{proof}

The complex domain is characterized by the following inequality opposite that for the real domain, i.e.,

\begin{equation}
\frac{1}{2} - v_2^T \overline{T}^2 v_2 + \frac{1}{4} \cos^2\theta(v_2^T \overline{T} v_2)^2 < 0.
\end{equation}

Recall that  $\theta = \arcsin(\sqrt{3}\tan(\frac{1}{3}\arcsin(\mu)))$ (Theorem~\ref{thm:medium_eigenvector_manifold}), which implies $\theta$ is monotonically increasing with respect to mode values $\mu$. Consider the range of $\theta$ which is $[0, \frac{\pi}{2}]$, $\cos^2\theta$ is a monotonically decreasing function with respect to $\mu$. That is, for $\mu_1<\mu_2$, their corresponding $\theta$'s satisfy $cos^2\theta_1>\cos^2\theta_2$.

Consequently, a unit vector $v_2$ that satisfies

\begin{equation}
\frac{1}{2} - v_2^T \overline{T}^2 v_2 + \frac{1}{4} \cos^2\theta_1(v_2^T \overline{T} v_2)^2 < 0
\end{equation}

\noindent must also satisfy

\begin{equation}
\frac{1}{2} - v_2^T \overline{T}^2 v_2 + \frac{1}{4} \cos^2\theta_2(v_2^T \overline{T} v_2)^2 < 0.
\end{equation}

\noindent That is, if $v_2$ is in the complex domain for $\mu_1$ then it must also reside in the complex domain $\mu_2$. The reverse is not true. Consequently, the complex domain for $\mu_1$ is a proper subset of that of $\mu_2$.

\end{proof}

\begin{lemma}
Given a 3D linear tensor field $T(x, y, z)=T_0+xT_x+yT_y+zT_z$ and a mode value $\mu$, the complex domain boundary corresponding to $\mu$ is parameterizable by $\alpha=v_2\overline{T}v_2$. \label{lemma:complex_boundary_parameterizable}
\end{lemma}

\begin{proof}

Let $\overline{\lambda}_1 \ge \overline{\lambda}_2 \ge \overline{\lambda}_3$ be the eigenvalues of $\overline{T}$ and $\overline{v}_1$, $\overline{v}_2$, and $\overline{v}_3$ be the corresponding eigenvectors. For convenience, we choose the coordinate system $\{\overline{v}_1, \overline{v}_2, \overline{v}_3\}$ in which $\overline{T}$ is a diagonal matrix. Let $v_2 = \begin{bmatrix} a \\ b \\ c \end{bmatrix}$ and $m = \begin{bmatrix} m_x \\ m_y \\ m_z \end{bmatrix} = \begin{bmatrix} a^2 \\ b^2 \\ c^2 \end{bmatrix}$. Equation~\ref{eq:real_boundary} implies
\begin{align}
	0
	&= 1 - 2 (m_x \overline{\lambda}_1^2 + m_y \overline{\lambda}_2^2 + m_z \overline{\lambda}_3^2) \nonumber\\
	&+ \frac{1}{2} \cos^2\theta (m_x \overline{\lambda}_1 + m_y \overline{\lambda}_2 + m_z \overline{\lambda}_3)^2.
	\label{eqn:real_complex_bound_m}
\end{align}

This equation's dependence on only $m_x=a^2$, $m_y=b^2$, and $m_z=c^2$ implies an eight-fold symmetry for the medium eigenvector manifold, and thus the complex domain boundary.
That is, if a unit vector $(a, b, c)$ is on the complex domain boundary, so is $(\pm a, \pm b, \pm c)$.

Therefore, we only need to parameterize the segment of the complex domain boundary where $a,b,c\ge 0$. The other seven segments can be parameterized in a similar fashion.

Notice that $m$ is always on the plane $m_x + m_y + m_z = \|v_2\|^2 = 1$. In the coordinate system

	\begin{align}
		\alpha &= \overline{\lambda}_1 m_x + \overline{\lambda}_2 m_y + \overline{\lambda}_3 m_z  \label{eq:complex_boundary_alpha} = v_2^T \overline{T} v_2 \\
		\beta  &= \overline{\lambda}_1^2 m_x + \overline{\lambda}_2^2 m_y + \overline{\lambda}_3^2 m_z \label{eq:complex_boundary_beta} = v_2^T \overline{T}^2 v_2
	\end{align}
	Equation~\ref{eqn:real_complex_bound_m} becomes
	\begin{align}
		0 &= 1 - 2 \beta + \frac{1}{2} \cos^2\theta \alpha^2
		\label{eqn:real_complex_bound_alpha}
	\end{align}
	which is the equation of a parabola that can be parameterized by $\alpha$. Each $\alpha$ gives one corresponding $\beta$ from Equation \ref{eqn:real_complex_bound_alpha} and thus one point on the segment of the complex domain boundary where $a,b,c\ge 0$.
\end{proof}

\begin{theorem}
Given a 3D linear tensor field $T(x, y, z)=T_0+xT_x+yT_y+zT_z$, let $\overline{\mu}$ is the mode of the characteristic tensor $\overline{T}$ and $\mu_0=\sqrt{1-\overline{\mu}^2}$. The topology of the generalized mode $\mu$ surface is a topological torus when $\mu>\mu_0$ and a topological double-torus when $\mu<\mu_0$.
\label{thm:bifucation}
\end{theorem}

\begin{proof}

We reuse the expressions $m_x$, $m_y$, $m_z$, $\alpha$, and $\beta$ from Lemma~\ref{lemma:complex_boundary_parameterizable}. Note that $m_x\ge 0$, $m_y\ge 0$, $m_z \ge 0$, and $m_x+m_y+m_z=1$. Therefore, points satisfying these conditions form an equilateral triangle in the plane $m_x+m_y+m_z=1$, which is illustrated in Figure~\ref{fig:complex_domain_bifurcation}.

Recall that $\overline{T}$ is a unit, traceless tensor with a non-positive determinant (Equation~\ref{T_bar_condition}). We have

\begin{eqnarray}
\label{eq:fact_1} \overline{\lambda}_1+\overline{\lambda}_2+\overline{\lambda}_3=0 \\
\label{eq:fact_2} \overline{\lambda}_1^2+\overline{\lambda}_2^2+\overline{\lambda}_3^2=1 \\
\label{eq:fact_3} \overline{\lambda}_1 \ge \overline{\lambda}_2 \ge 0 \ge \overline{\lambda}_3
\end{eqnarray}

Plugging in $\overline{\lambda}_3 = -\overline{\lambda}_1 - \overline{\lambda}_2$ into Equation~\ref{eq:fact_2}, we have $\overline{\lambda}_1^2+\overline{\lambda}_1\overline{\lambda}_2+\overline{\lambda}_2^2 = \frac{1}{2}$. This implies that $\overline{\lambda}_1^2+\overline{\lambda}_2^2 <  \frac{1}{2}$. Consequently, $\overline{\lambda}_1^2 < \frac{1}{2}$, $\overline{\lambda}_2^2 < \frac{1}{2}$, and $\overline{\lambda}_3^2 > \frac{1}{2}$.

For degenerate curves, i.e.\ $\mu=1$, $\cos\theta=0$ and Equation~\ref{eqn:real_complex_bound_alpha} reduces to

	\begin{align}
		0 &= 1 - 2 \beta,
		\label{eqn:real_complex_bound_alpha_mode_1}
	\end{align}

\noindent which further reduces to $1-2\overline{\lambda}_i^2$ at the corners of the triangle. Consequently, this function is negative at $(0, 0, 1)$ (complex domain) and positive at $(1, 0, 0)$ and $(0, 1, 0)$ (real domain). Because the function is linear, its zeroth levelset (complex domain boundary) must intersect each of the edges $m_x=0$ and $m_y=0$ exactly once. The subtriangle formed by the two intersection points and $(0, 0, 1)$ is the complex domain corresponding to $\mu=1$.

As $\mu$ decreases, the complex domain (grey regions in Figure~\ref{fig:complex_domain_bifurcation}) reduce in size (Lemma~\ref{lemma:complex_boundary_parameterizable}), and the parabola characterized by Equation~\ref{eqn:real_complex_bound_alpha} moves lower while still intersecting $m_x=0$ and $m_y=0$ at one point each (the yellow dots on the two edges in (a)). Due to the aforementioned eight-fold symmetry in the medium eigenvector manifold (Lemma~\ref{lemma:complex_boundary_parameterizable}), this segment of the parabola corresponds to eight segments that constitute the complex domain boundary in the medium eigenvector manifold. Due to the antipodal symmetry, we only consider the four segments where $c>0$. Since the original parabolic segment touches both the $m_x=0$ and $m_y=0$ edges, these four segments will be connected, forming a single loop (Figure~\ref{fig:complex_para}: the outermost loop). Note that due to the antipodal symmetry, the other four segments in the medium eigenvector manifold ($c<0$) also form a single loop. Moreover, this loop is to be identified with the loop where $c>0$. The real domain in the medium eigenvector manifold is thus the part of the sphere between these two loops. Gluing the two loops based on the antipodal symmetry results in a space without a boundary, the torus, which is homeomorphic to the generalized mode $\mu$ surface.

\begin{figure}
\subfloat[][Before bifurcation]{\includegraphics[width=1.2in]{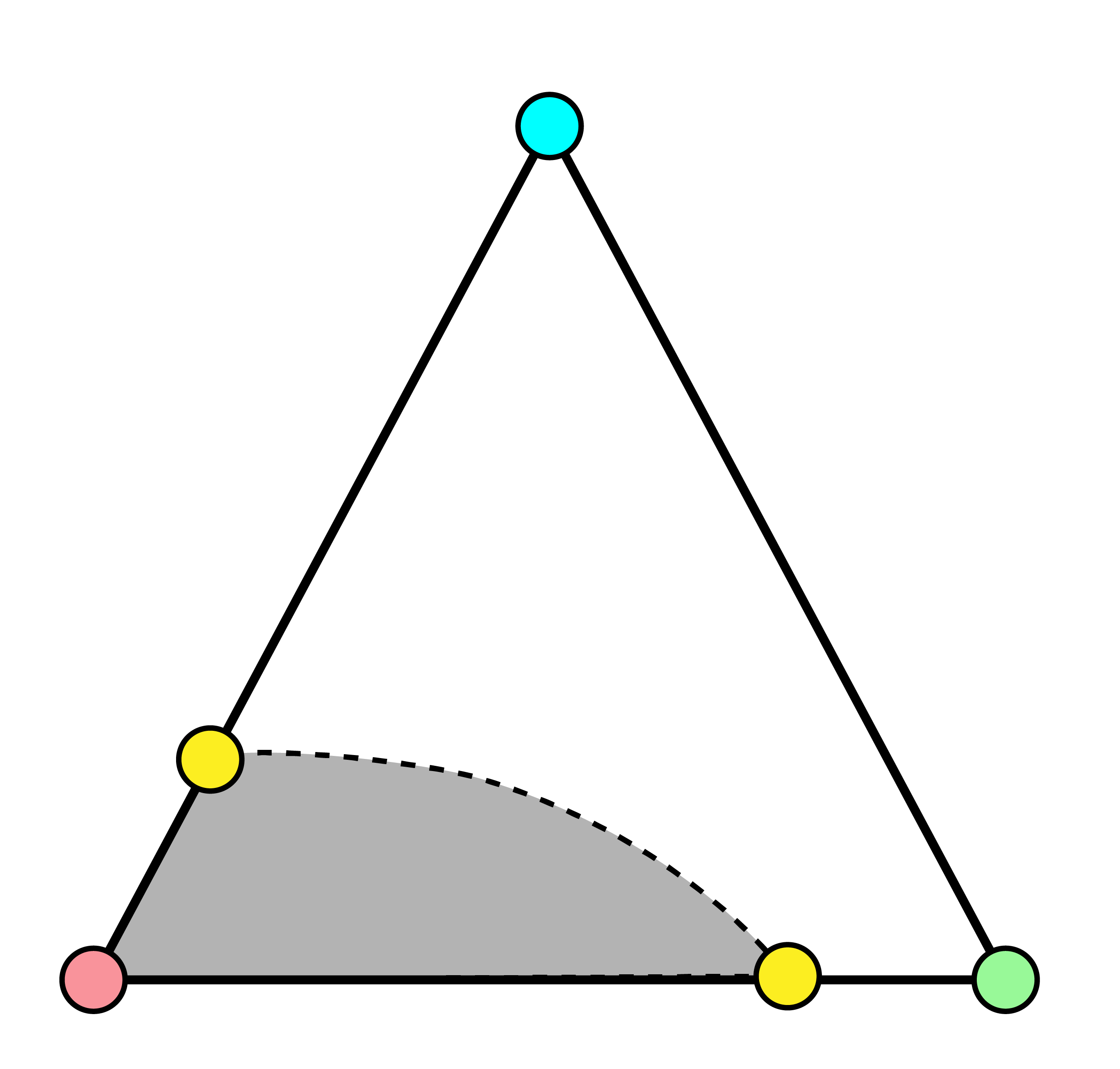}}
\subfloat[][Bifurcation]{\includegraphics[width=1.2in]{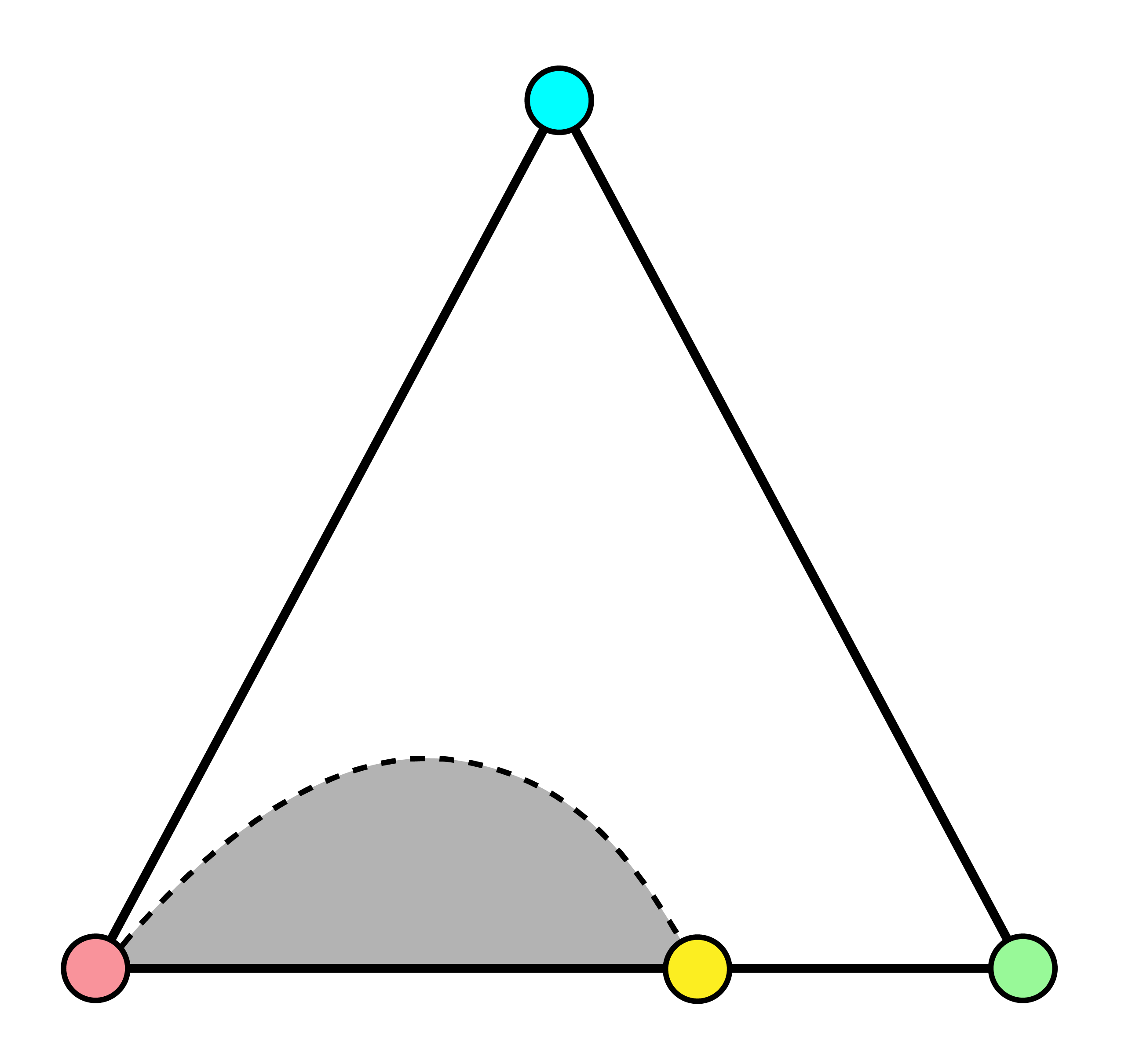}}
\subfloat[][After bifurcation]{\includegraphics[width=1.2in]{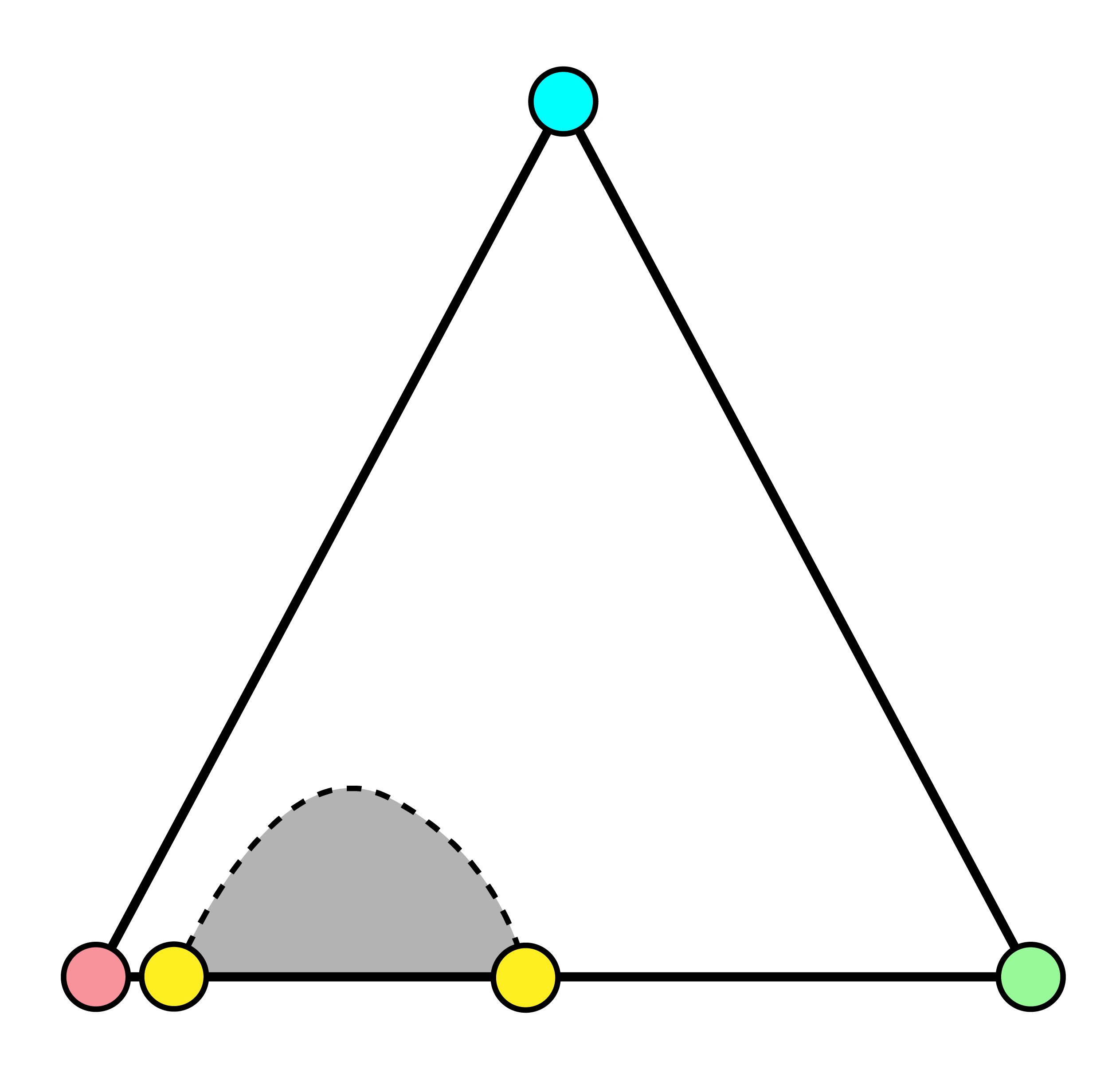}}
\caption{In the triangle bounded by $m_x=0$, $m_y=0$, and $m_z=0$ in the plane $m_x+m_y+m_z=1$, the corners of the triangle are as follows: $(0, 0, 1)$ (red),  $(1, 0, 0)$ (green), and $(0, 1, 0)$ (cyan). When $\mu$ decreases,  the complex domain (grey regions) reduce in sizes. The complex domain boundary initially intersects each of the $m_x=0$ and $m_y=0$ edges once (a). After the bifurcation point (b), the complex domain only intersects the $m_y=0$ edge (two points). The topology of the generalized mode $\mu$ surface depends on whether there is a solution when $m_x=0$.} \label{fig:complex_domain_bifurcation}
\end{figure}

The above situation changes when the parabola intersects the corner of $m_x=0$ and $m_y=0$ plus an additional point on $m_y=0$ (Figure~\ref{fig:complex_domain_bifurcation} (b)). It can be shown that this occurs when $\mu=\sqrt{1-\overline{\mu}^2}$, which is the bifurcation point. When this happens, the four segments where $c>0$ form a figure-eight (Figure~\ref{fig:complex_para}: the loop between the outermost and innermost loops). The real domain in the medium eigenvector manifold is thus the part of the sphere between this loop and its antipodal image. Gluing the two figure-eights results in a surface that has a non-manifold point corresponding to the center of the figure-eight. That is, the generalized mode $\mu$ surface in this case is a non-manifold surface.

Decreasing $\mu$ further, the parabolic segment only intersects the $m_y=0$ at two points (Figure~\ref{fig:complex_domain_bifurcation} (c)). This implies that the four corresponding segments ($c>0$) in the medium eigenvector manifold form two loops (Figure~\ref{fig:complex_para}: the innermost loop). Similarly, the four segments ($c<0$) also form two loops. The real domain is the part of the sphere between the four loops. Gluing the four loops pairwise according to the antipodal symmetry results in a sphere with two handles attached, i.e.\ double-torus. This is the topology of generalized mode $\mu$ surface when $\mu<\mu_0$.

When $\mu=0$, it is straightforward to verify that the parabolic segment degenerates to a single point, and the aforementioned four loops shrink to two pairs of antipodal points:
\begin{equation}
\begin{pmatrix} \pm \sqrt{\frac{\overline{\lambda}_1-\overline{\lambda}_2}{\overline{\lambda}_1-\overline{\lambda}_3}} & 0 & \pm \sqrt{\frac{\overline{\lambda}_2-\overline{\lambda}_3}{\overline{\lambda}_1-\overline{\lambda}_3}} \end{pmatrix}  \label{eq:singularities}
\end{equation}

Each of the pairs corresponds to one of the two singularities in the medium eigenvector manifold for neutral surfaces~\cite{Roy:18}.
\end{proof}

\begin{theorem}
Given a 3D linear tensor field $T(x, y, z)=T_0+xT_x+yT_y+zT_z$ and a plane $P$, the critical points of the mode function on the plane $P$ consist of at most four extrema.
\label{thm:at_most_four_extrema}
\end{theorem}

\begin{proof}

On the plane $P$, the tensor field $T(x, y, z)$ is still a linear tensor field. Therefore, without the loss of generality, we assume the plane $P$ to be the $XY$ plane. If this is not the case, we can simply perform a space transformation so that $P$ becomes the $XY$ plane under the new coordinate system.

In the $XY$ plane, the tensor field has the form $T(x, y)=T_0+xT_x+yT_y$. We define a map $\chi(x, y) = \frac{T(x, y)}{\| T(x, y) \|}$ from the plane $P$ to the set of unit tensors. Since the tensors on the plane is a combination of $T_0$, $T_x$, and $T_y$, the image of this map together with its negation forms a two-dimensional sphere in the linear subspace spanned by $T_0$, $T_x$, and $T_y$. The map $\chi$ is injective because if a tensor has already appeared in a plane, its multiples cannot. Furthermore, $\chi$ is locally surjective. Therefore, the set of critical points of the mode function in the plane $P$ has a one-to-one correspondence to the critical points of the mode function on the aforementioned two-dimensional sphere.

The mode of a unit tensor $T$ is $3 \sqrt{6} \det(T)$, which means that the critical points of the mode function on the sphere can be found by computing the critical points of the determinant function $\det(T)$ on the sphere. For convenience, we choose $T_1, T_2, T_3$ to be a orthonormal basis for the space spanned by $T_0$, $T_x$, and $T_y$.

The critical points of a function defined on the sphere are the points where its gradient is colinear to the sphere's normal. Thus, the critical points of the determinant function satisfy

\begin{eqnarray}
\nabla \det(u T_1 + v T_2 + w T_3) \times \begin{bmatrix} u \\ v \\ w \end{bmatrix} = 0  \label{eq:critical_determinant}\\
u^2 + v^2 + w^2 = 1.
\end{eqnarray}

Let \begin{equation}
\nabla \det(u T_1 + v T_2 + w T_3) = \begin{pmatrix} f(u, v, w) \\ g(u, v, w) \\ h(u, v, w) \end{pmatrix}
\end{equation}
\noindent where $f(u, v, w)$, $g(u, v, w)$, and $h(u, v, w)$ are quadratic polynomials. Thus, Equation~\ref{eq:critical_determinant} consists of the following three cubic equations:

\begin{eqnarray}
vh(u, v, w)=wg(u, v, w) \nonumber \\
wf(u, v, w)=uh(u, v, w) \nonumber \\
ug(u, v, w)=vf(u, v, w).  \label{eq:critical_determinant_three}
\end{eqnarray}

Together with the condition $u^2 + v^2 + w^2 = 1$, we have an over-determined system of three cubic equations and one quadratic equation. Fortunately, the three cubic equations (Equation~\ref{eq:critical_determinant_three}) are almost redundant. To see this, notice that multiplying the first two of these equations gives rises to

\begin{equation}
vwf(u,v,w)h(u, v, w)=uwg(u,v,w)h(u, v, w)
\end{equation}

When $w\ne0$ and $h(u,v,w)\ne0$, the above equation reduces to the third cubic equation $ug(u, v, w)=vf(u, v, w)$. This shows the redundancy of over-determined system.

After removing the redundant third equation, we get the system
\begin{eqnarray}
vh(u, v, w)=wg(u, v, w) \label{eq:critical_ax}\\
wf(u, v, w)=uh(u, v, w) \label{eq:critical_ay}\\
u^2 + v^2 + w^2 = 1, \label{eq:critical_sphere}
\end{eqnarray}
which has up to $18$ complex solutions based on B\'ezout's theorem~\cite{Fulton:74}. However, there are some spurious solutions. Note that for $vwf(u, v, w)h(u, v, w)=uwg(u, v, w)h(u, v, w)$ to be equivalent to $ug(u, v, w)=vf(u, v, w)$, we require that $w\ne0$ and $h(u,v,w)\ne0$. Assuming that $w=0$, then we must have $vh(u, v, w)=0=uh(u, v, w)$. Since $u$ and $v$ cannot be both zeros (or the vector $(u, v, w)$ is the zero vector), we must have $h(u, v, w)=0$.

Solutions satisfying $h(u, v, 0)=0$ and $u^2+v^2=1$ have four complex solutions, which are the spurious solutions to our system. Consequently, we have up to $18-4=14$ complex solutions, thus 14 real solutions at most. Note that if $(u, v, w)$ is a solution, so is $(-u, -v, -w)$. Moreover, they represent the same tensor. Consequently, there are only up to seven real solutions, i.e.\ up to seven critical points in the mode function in the plane.

The Euler Characteristic of a sphere is two~\cite{armstrong1979basic}, which means that the number of extrema is always two more than the number of saddles according to Morse theory~\cite{matsumoto2002introduction}. Together with the fact that there are at most $14$ critical points, we can conclude that there are at most eight extrema on the sphere. Since the antipodal points on the sphere give the same point on the plane, we have at most four extrema of the mode function on the plane $P$.

\end{proof}

\section{One More Application Scenario}
\label{sec:additional_example}

\begin{figure}[!t]
\centering%
\subfloat[][]{\includegraphics[height=0.5in]{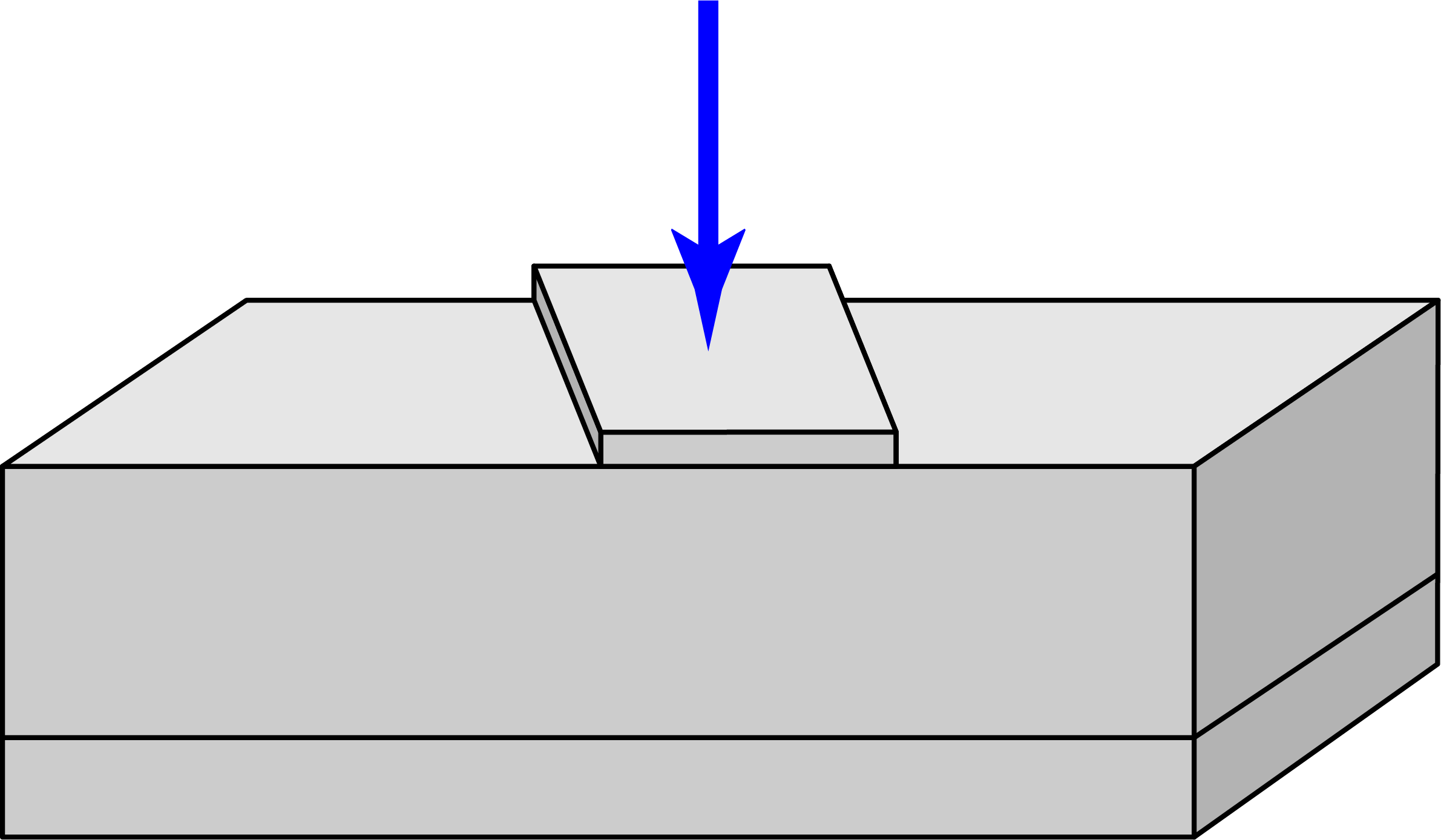}}
\caption[]{Another scenario of compression in a solid block.}
\label{fig:geometry2}
\end{figure}

We provide one more scenario of compression on a block (Figure~\ref{fig:geometry2}) in which the compressive load is misaligned with the natural geometric orientation of the block due to the addition of a slanted slab on the top of the block. This slanted slab presses down on the block.  At the bottom of the block, we also add a full-length layer with a more rigid material.  This set-up aims to provide a simplified representation of a tire tread block anchored on some steel belts that hold the tire together while driving on hard pavement. Nonetheless, this type of boundary condition is not limited to tire design. It can be due to misuse of the structure. In particular, both cases exist on tires, i.e.\ some designs angle tread blocks from the circumferential direction of the tire to facilitate rainwater drainage while irregular wear such as tire cupping induces uneven tread block wear on the shoulders of the tire.

The neutral surface of the stress is shown in Figure~\ref{fig:24_frontAndBackView} (a), which, interestingly, reflects the existence of the slab. In addition, there are two sheets of neutral surfaces attached to the top face, two sheets attached to the sidewalls, and two tubes that connect the front and back faces. The location of these surfaces which indicate pure shear regions match our expectation due to our design of the boundary conditions. In addition, other mode surfaces can also provide meaningful insight. For example, we show the generalized mode $0.4$ surface in Figure~\ref{fig:24_frontAndBackView} (b). While the positive part of this surface (teal) and the negative part of the surface (gold) both approximate the neutral surface (Figure~\ref{fig:24_frontAndBackView} (a)), the similarity stops there. The positive part of the mode surface contains additional sheets that are not present in the negative part of the mode surface. This indicates that despite the compression-dominant boundary condition, extension can exist in the middle of the volume. Seeing this positive mode surface confirms the bulging effect from compression.
\begin{figure}[!tb]
  \centering
\subfloat[][Neutral Surface]{\includegraphics[width=2.9in]
{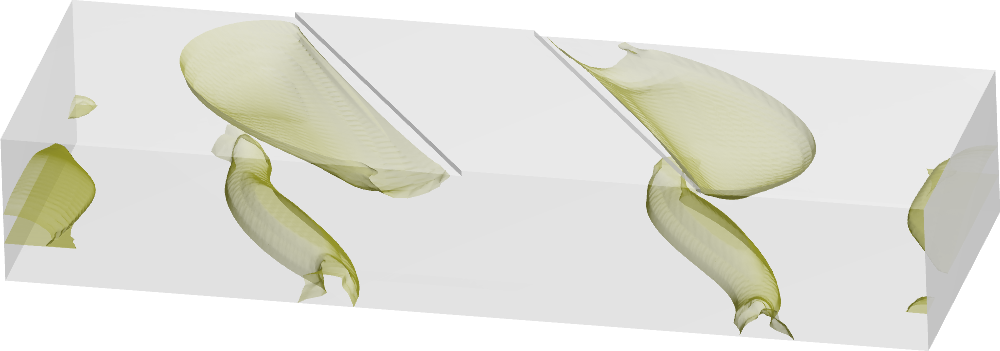}}\\
\subfloat[][$\mu=\pm 0.4$]{\includegraphics[width=2.9in]
{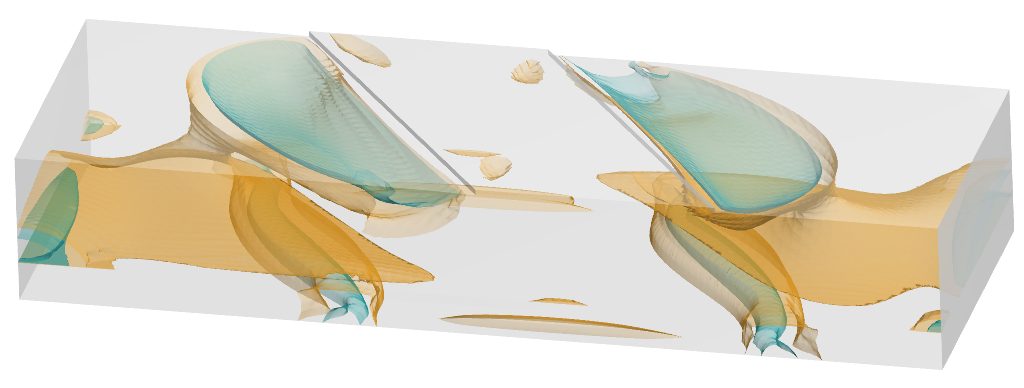}}
\caption{Mode surfaces of stress tensor fields of a rectangular block being pushed down by a slanted slab. }\label{fig:24_frontAndBackView}
\end{figure}

This visualization can potentially augment the definition of uniaxial compression, the most fundamental description of compression numerically, in which the compressive forces are not aligned with the underlying geometry.

\end{document}